\emergencystretch=\maxdimen
\documentclass[12pt,a4paper]{article} 

\usepackage[utf8]{inputenc} 
\usepackage[english]{babel}
\usepackage{mathtools,amsmath,amsfonts,amsthm,amssymb} 
\usepackage{tikz-cd} 
\usepackage{adjustbox} 
\newtheorem{definition}{Definition}
\newtheorem{proposition}{Proposition}
\numberwithin{equation}{section} 
\usepackage{authblk} 
\usepackage{indentfirst} 
\usepackage{sectsty} 
\allsectionsfont{\centering \normalfont \scshape} 
\usepackage[margin=3cm]{geometry} 
\usepackage{braket}

\def\eprint#1{\texttt{arXiv:#1}}

\title{Airy structures for semisimple Lie algebras}
\author{Leszek Hadasz, Błażej Ruba}
\affil{\textit{Institute of Physics, Jagiellonian University}}
\affil{\textit{prof. Łojasiewicza 11, 30-348 Kraków, Poland}}
\date{\today}

\begin{document}
\frenchspacing

\maketitle

\begin{abstract}
We give a complete classification of Airy structures for finite-dimensional simple Lie algebras over $\mathbb C$, and to some extent also over $\mathbb R$, up to isomorphisms and gauge transformations. The result is that the only algebras of this type which admit any Airy structures are $\mathfrak{sl}_2$, $\mathfrak{sp}_4$ and $\mathfrak{sp}_{10}$. Among these, each admits exactly two non-equivalent Airy structures. Our~methods apply directly also to semisimple Lie algebras. In this case it~turns out that the number of non-equivalent Airy structures is countably infinite. We have derived a number of interesting properties of these Airy structures and constructed many examples. Techniques used to~derive our results may be described, broadly speaking, as~an~application of~representation theory in semiclassical analysis.
\end{abstract}

\tableofcontents

\section{Introduction}
Quantum Airy structure is a set of differential operators of the form\footnote{Repeated indices are always summed over.}
\begin{equation}
\label{Airy:definition}
L_i = \hbar \partial_i - \frac{1}{2} A_{ijk} x^j x^k - \hbar B_{ij}^k x^j \partial_k - \frac{\hbar^2}{2} C_i^{jk} \partial_j \partial_k - \hbar D_i, \quad i \in \{ 1, ..., n \},
\end{equation}
spanning a Lie algebra $\mathfrak{g}$ with structure constants $f_{ij}^k:$
\begin{equation}
\hbar^{-1}[L_i, L_j ] = f_{ij}^k L_k.
\end{equation}
Airy structures were introduced in \cite{Kontsevich} as a reformulation and generalization of~a~system of recursive equations, referred to as the Chekhov-Eynard-Orantin topological recursion \cite{CE,EO,EO2}.
Formulated originally in the language of matrix motels, the CEO topological recursion
can be rephrased more abstractly as a~procedure which assigns invariants to spectral curves, i.e.~Riemann surfaces equipped with certain additional geometric structre \cite{EO,EO2}. This turned out to be useful in the study of Hurwitz numbers \cite{Bouchard:2007hi,BSLM,Eynard:2009xe,Do:2012udk,Alexandrov:2018ncq}, computation of~Gromov-Witten invariants \cite{EO3}, in knot theory \cite{Dijkgraaf:2010ur,Borot:2012cw}, integrable systems \cite{TR_Integrable1, TR_Integrable2} and topological quantum field theories \cite{TR_TQFT}. Furthermore it is connected with the subject of quantum curves \cite{Gukov:2011qp,Mulase:2012tm,Bouchard:2016obz}.

It is thus conceivable that results concerning Airy structures (and their supersymmetric generalizations \cite{Bouchard:2019uhx},
related to supereigenvalue models and the corresponding topological recursion
\cite{Bouchard:2018anp,Ciosmak:2017ofd,Ciosmak:2016wpx,Ciosmak:2017omd,Osuga:2019uqc}) may
find applications in some of the subjects listed above.

Every quantum Airy structure admits \cite{Kontsevich, Analyticity} a unique ``free energy'' $F$, which is a series in $\hbar$ and $x^i$ satisfying differential equations
\begin{equation}
L_i \cdot e^{\hbar^{-1} F}=0
\end{equation}
and initial conditions $F(0, \hbar) = \partial_i F(0,0) = \partial_i \partial_j F(0,0)=0$. Thus the corresponding partition function $Z = e^{\hbar^{-1}F}$ may be viewed as WKB wave function\footnote{In certain cases this wave function really is the partition function of some system.} of a quantum system whose symmetry is generated by hamiltonians $L_i$. 

%

One important question about any class of mathematical objects is the classification problem, which asks for a complete list of all (up to suitably defined equivalence) objects satisfying the pertinent axioms. It is unlikely that such list of~all Airy structures could ever be obtained. However, one may still hope to~classify some special classes of Airy structures. In \cite{ABCD} study of this problem was initiated for Airy structures for which the Lie algebra $\mathfrak g$ is~finite-dimensional and simple. Somewhat surprisingly, only one example was found.~This suggests that assumption of simplicity imposes very strong con\-straints. Indeed, in this work we~provide a solution to this classification problem. It~turns out that there exist precisely six inequivalent Airy structures, two for each of $\mathfrak{sl}_2$, $\mathfrak{sp}_4$ and $\mathfrak{sp}_{10}$. We~construct these Airy structures explicitly. More detailed summary is given in the Subsection \ref{ss:summary}. Methods developed in order to obtain these results are of~independent interest, because they apply also to the case of~semisimple $\mathfrak g$. In~this more general case classification program is not finished, but~significant progress has been made in this direction.

The main ideas and methods applied to constrain and construct Airy structures can be summarized as follows. 
Since for the class of Airy structures under consideration there exists an unambiguous procedure of quantization, no~information is lost by working with the classical hamiltonians instead of~directly with quantum operators (\ref{Airy:definition}).  
As explained in \cite{Kontsevich}, any such classical Airy structure may be obtained by expressing the moment map $\ell$ of a hamiltonian action of $\mathfrak g$ on some affine space of dimension $2 \dim \mathfrak g$ in standard coordinates centered at a regular point of $\ell^{-1}(0)$. For semisimple Lie algebras all affine representations are actually linear and completely classified. 
The~part which is not known, to the best of our knowledge, is for which representations the locus $\ell^{-1}(0)$, also called the~characteristic variety, has any regular points.
To answer this question we use the fact that the set of regular points of $\ell^{-1}(0)$ is a cone with a~locally transitive action of a complex Lie group with Lie algebra $\mathfrak g$. 
This implies that for any of its points $\Omega$ one may find a unique element $J$ of the algebra $\mathfrak g$ which is tangent to the ray of $\Omega$. It is possible to describe many properties of $J$, including its spectrum. Having obtained that, we proceed to the classification. Instead of~looking directly for $\Omega$, we find all possible forms of $J$. Once some admissible $J$ is~found, element $\Omega$ is obtained by solving the eigenvalue equation $J \Omega = \Omega$.


The paper is organized as follows. In Section \ref{Preliminaries} we recall some basic facts about Airy structures. It contains no new results, but it introduces the language used afterwards. In Section \ref{sec:semisimple_general} we discuss properties of the characteristic variety and of~elements $\Omega$ and $J$.~These are the main conceptual ingredients of the classification program. Section \ref{sec:automorphisms} concerns automorphisms and real forms of Airy structures. In~Section \ref{sec:simple_class} we perform explicit calculations, which culminate in the promised list of~Airy structures for simple Lie algebras.~Examples of application of~our~formalism to~semisimple Lie algebras are presented in Section \ref{sec:ss_examples}. We~summarize by mentioning possible future directions in Section \ref{sec:conclusions}. For~convenience of the reader we collect some background material in appendices. Appenix \ref{app:cohom} introduces in~an~elementary way Lie algebra cohomology groups, which are used throughout the text. In~the~Appendix \ref{app:ss_regular} we~recall the notions of~semisimple and regular elements of~a~semisimple Lie algebra. Appendix \ref{app:invariant_poly} contains a brief discussion of invariant polynomials on semisimple Lie algebras. We~find relations between invariant polynomials of various types for the~Lie~algebra $\mathfrak{sp}_{10}$, which is used to find the~element $J$ in this case.

\section{Preliminaries}\label{Preliminaries}

 In general it is necessary to~impose additional finitness conditions on the tensors $A,B,C,D$ and $f$ appearing in (\ref{Airy:definition}). These~are automatically satisfied if $ \dim \mathfrak g$ is~finite, which we assume from now~on. We~work over the field~$\mathbb C$, but our results are relevant also for Airy structures over~$\mathbb R$. Indeed, every real Airy structure admits a~natural complexification. Furthermore, we discuss the concept of real forms of Airy structures in Section \ref{sec:automorphisms}.

The classical limit of a quantum Airy structure is the set of hamiltonians
\begin{equation}
\ell_i = y_i - \frac{1}{2} A_{ijk} x^j x^k - B_{ij}^k x^j y_k - \frac{1}{2} C_{i}^{jk} y_j y_k.
\label{eq:cl_hamiltonians}
\end{equation}
They satisfy relations
\begin{equation}
\{ \ell_i, \ell_j \} = f_{ij}^k \ell_k
\label{eq:cl_commutator}
\end{equation}
with respect to the Poisson bracket defined by
\begin{equation}
\{ f, g \} = \frac{\partial f}{\partial y_i} \frac{\partial g}{\partial x^i} - \frac{\partial f}{\partial x^i} \frac{\partial g}{\partial y_i}.
\end{equation}
Classical Airy structure may be defined as a set of hamiltonians of the form (\ref{eq:cl_hamiltonians}) subject to relations (\ref{eq:cl_commutator}). Every classical Airy structure may be quantized by putting $D_i = \frac{1}{2} B_{ij}^j + \delta_i$ with any $\delta$ satisfying $f_{ij}^k \delta_k=0$. Thus the set of quantizations of a given classical Airy structure may be identified with the vector space $H^0(\mathfrak g, \mathfrak g^*)$. Choice $\delta=0$ corresponds to Weyl quantization. This description reduces the classification of Airy structures for a given Lie algebra to the study of their classical versions.

Classical Airy structures have a transparent geometric interpretation. Consider the~common zero locus of hamiltonians $\ell_i$,
\begin{equation}
\Sigma = \{ (x,y) \in \mathbb C^{n} \times \mathbb C^n | \ell_1(x,y)=...=\ell_n(x,y)=0 \}
\end{equation}
and its Zariski open subset
\begin{equation}
\Sigma_s = \{ q \in \Sigma | \left. d \ell_1 \wedge ... \wedge d \ell_n \right|_{q} \neq 0 \}.
\end{equation}
Postulated form of $\ell_i$ implies that the origin belongs to $\Sigma_s$. Conversely, given a~set of at most quadratic hamiltonians on $\mathbb C^{2n}$ satisfying (\ref{eq:cl_commutator}), define $\Sigma$ and $\Sigma_s$ as~above. Then $\Sigma_s$ is a Lagrangian submanifold. For any $\Omega \in \Sigma_s$ one can choose a symplectic affine coordinate chart centered at~$\Omega$ in~which $\ell_i$ take the form (\ref{eq:cl_hamiltonians}). Coordinate systems with desired properties are in one-to-one correspondence with Lagrangian complements of $T_{\Omega} \Sigma_s$ in $\mathbb C^{2n}$. Hamiltonians corresponding to different complements are related by a~change of coordinates
\begin{equation}
y_i \mapsto y_i, \quad \quad x^i \mapsto x^i + s^{ij} y_j,
\label{eq:gauge_trans}
\end{equation}
where $s$ is a symmetric matrix. Maps of this form are called gauge transformations. There exists an analogous notion for quantum Airy structures \cite{Kontsevich}. Transformation law for the associated partition functions is also known \cite{ABCD,Bouchard:2019uhx}. One should not fall under the impression that choice of the Lagrangian complement is completely irrelevant: some choices lead to much simpler partition functions, and transforming the partition function to other gauges is not trivial. This description allows one to~define classical Airy structures in a~way not referring to coordinates. We~introduce also the concept of an Airy data, which can be thought of as equivalence classes of~Airy structures up to gauge transformations.





\begin{definition} \label{def:cl_Airy}
Classical Airy structure is a quadruple $(\mathfrak g, W, \Omega , V)$, where
\begin{enumerate}
\item $\mathfrak g$ is a Lie algebra of dimension $n$.
\item $W$ is an affine space of dimension $2n$ equipped with a translation invariant symplectic form $\omega$ and a $\mathfrak g$-action $\xi : \mathfrak g \to \Gamma(TW)$ on $W$ which is hamiltonian with at most quadratic moment map $\ell : W \to \mathfrak g^*$.
\item $\Omega$ is an element of $\Sigma_s = \{ q \in \ell^{-1}(0) | \left. d \ell \right|_q : T_q W \to \mathfrak g^* \text{ has rank } n \} $.
\item $V$ is a Lagrangian complement of $T_{\Omega} \Sigma_s$ in $W$.
\end{enumerate}
Triple $(\mathfrak g, W, \Omega)$ satisfying points $1-3$ of the above list is called an Airy datum. We~say that Airy datum is nontrivial if $n>0$.
\end{definition}

We will frequently use the following characterization of the set $\Sigma_s$.

\begin{proposition} \label{prop:Omega_criterion}
Let $(\mathfrak g, W)$ be as in the Definition \ref{def:cl_Airy}. Suppose that $\mathfrak g = [\mathfrak g, \mathfrak g]$. Then
\begin{equation}
\Sigma_s = \{ \Omega \in W |  \{ \left. \xi(T) \right|_{\Omega} \}_{T \in \mathfrak g} \text{ is Lagrangian}  \}.
\end{equation}
\end{proposition}
\begin{proof}
$\subseteq$ : Follows from the preceding discussion (for any $\mathfrak g$). \\
$\supseteq$ : If $\{ \left. \xi(T) \right|_{\Omega} \}_{T \in \mathfrak g}$ is Lagrangian, $\left. d \ell \right|_{\Omega}$ has rank $n$. Furthermore, we have
\begin{equation}
\ell(\Omega)([T,S]) = \left. \omega( \xi(T) ,  \xi(S) ) \right|_{\Omega}=0.
\end{equation}
Thus $\ell(\Omega)=\left. \ell(\Omega) \right|_{[\mathfrak g, \mathfrak g]}=0$.
\end{proof}

\begin{definition} \label{def:Airy_iso}
Homomorphism of Airy structures $(\mathfrak g, W, \Omega, V) \to (\mathfrak g', W',\Omega', V')$ is~a~pair $(\phi, f)$, where $\phi : \mathfrak g \to \mathfrak g'$ is a homomorphism of Lie algebras and $f$ is~an~affine Poisson map $W \to W'$, subject to the following conditions:
\begin{enumerate}
\item $\ell  = \phi^t \circ \ell' \circ f$, where $\ell'$ is the moment map of $W'$ and $\phi^t$ is~the transpose of $\phi$.
\item $f(\Omega)=\Omega'$.
\item $\left. df \right|_{\Omega}(V) \subseteq V'$.
\end{enumerate}
If $(\phi, f)$ satisfies only conditions $1$ and $2$, we say that it is a homomorphism of~Airy data. In other words, the following diagram is required to be commutative:
\adjustbox{scale=1.5,center}{%
\begin{tikzcd}[contains/.style = {draw=none,"\in" description,sloped}]
\Omega \arrow[contains]{r} \arrow[mapsto]{d} & W \arrow{r}{\ell} \arrow{d}{f} & \mathfrak g^* \\
\Omega' \arrow[contains]{r} & W' \arrow{r}{\ell'} & \mathfrak g'^* \arrow{u}{\phi^t}
\end{tikzcd}
}
\end{definition}

\begin{proposition} \label{prop:Airy_morph_basic_prop}
If $(\phi,f) : (\mathfrak g, W, \Omega) \to (\mathfrak g', W',\Omega')$ is a homomorphism of Airy data, then $\phi$ and $f$ are surjective. Moreover we have
\begin{equation}
\forall q \in W \ \forall T \in \mathfrak g \ \left. df \right|_q (\xi(T)) = \xi'(\phi(T)),
\label{eq:Airy_morphism_intertwines}
\end{equation}
where $\xi'$ is the $\mathfrak g'$-action on $W'$. In particular $\phi$ is uniquely determined by $f$.
\end{proposition}
\begin{proof}
$f$ is an affine Poisson map, so it is surjective. In particular $\dim(\mathfrak g) \geq  \dim(\mathfrak g')$. Evaluating the differential of $\ell$ at $\Omega$ we get
\begin{equation}
\left. d \ell \right|_{\Omega} = \phi^t \circ \left. d \ell' \right|_{\Omega'} \circ \left. df \right|_{\Omega}.
\label{eq:Airy_morph_intertwines}
\end{equation}
Since the rank of $\left. d \ell \right|_{\Omega}$ is equal to $\dim (\mathfrak g)$, the rank of $\phi^t$ is at least $\dim(\mathfrak g)$. Now~choose $T \in \mathfrak g$, $q \in W$ and a holomorphic function $g$ on $W'$. Equation $\ell = \phi^t \circ \ell' \circ f$ implies that we have $\xi(T)(f^* g') = f^* \left( \xi'(\phi(T))(g) \right)$, or in other words
\begin{equation}
\left. dg \right|_{f(q)} \left( \xi'(\phi(T))-\left. df \right|_{q} \xi(T) \right)=0.
\end{equation}
Since $g$ was arbitrary, formula (\ref{eq:Airy_morph_intertwines}) follows. The last statement is a consequence of the formula (\ref{eq:Airy_morph_intertwines}) and the fact that the $\mathfrak g'$-action $\xi'$ is faithful.
\end{proof}

In view of the Proposition \ref{prop:Airy_morph_basic_prop}, one could abuse notation and refer to $f$ itself as a~morphism $(\mathfrak g, W , \Omega) \to (\mathfrak g', W', \Omega')$. We choose not to do so, because it is unclear how to rephrase point $1$ in the Definition \ref{def:Airy_iso} without referring to $\phi$.



Let $(\mathfrak g, W, \Omega)$ be an Airy datum and let $G$ be a simply-connected Lie group with Lie algebra $\mathfrak g$. The $\mathfrak g$-action on $W$ exponentiates to an affine action of $G$, which preserves $\ell^{-1}(0)$ and $\Sigma_s$. By the Jacobian criterion, $\Sigma_s$ is a~nonsingular subvariety of~$W$ of dimension $n$. In particular it is a complex manifold with finitely many connected components. The $G$-orbits in $\Sigma_s$ are open in $\Sigma_s$, so~they coincide with the connected components. For any $q \in \Sigma_s$, let~$\mathrm{Stab}(q) = \{ g \in G | g \cdot q = q \}$ and $\mathrm{Orb}(q) = \{ g \cdot q \in \Sigma_s | g \in G \}$. Mapping $G \ni g \mapsto g \cdot q \in \mathrm{Orb}(q)$ is a universal cover, with fiber $\pi_1(\mathrm{Orb}(q)) \cong \mathrm{Stab}(q)$. This means that to get a~rather complete picture of the topology of $\Sigma_s$, it is sufficient to find the connected components and compute the corresponding stabilizers. This task is relevant for the classification program for several reasons. If $(\mathfrak g, W, \Omega')$ is another Airy datum with $\Omega' \in \mathrm{Orb}(\Omega)$, there exists $g \in G$ such that $g \cdot \Omega=\Omega'$. Then $(\mathrm{Ad}_g,g)$ is an~isomorphism $(\mathfrak g, W, \Omega) \to (\mathfrak g, W, \Omega')$. Secondly, group $\mathrm{Stab}(\Omega)$ is an invariant of Airy data. It~allows to distinguish non-isomorphic Airy data with isomorphic $\mathfrak g$ and $W$. Finally, Airy structure and its~partition function may be regarded as a~quantization of~$\mathrm{Orb}(\Omega)$. Thus it is desirable to know what these spaces look like.

We finish this section with an elementary discussion of products of Airy data. We~remark that Airy datum constructed in the Definition \ref{def:Airy_product} is indeed a product (with obvious projection maps) in the category of Airy data. There is also an~analogous notion for classical and quantum Airy structures, but we shall not use it.

\begin{definition} \label{def:Airy_product}
Given two Airy data $(\mathfrak g_i, W_i, \Omega_i)$, $i=1,2$ we define the product
\begin{equation}
(\mathfrak g_1, W_1, \Omega_1) \times (\mathfrak g_2, W_2, \Omega_2) = (\mathfrak g_1 \times \mathfrak g_2, W_1 \times W_2, (\Omega_1,\Omega_2)),
\end{equation}
with the moment map on $W_1 \times W_2$ given by $\ell (q_1,q_2) (T_1, T_2) = \ell_1(q_1)(T_1) + \ell_2(q_2)(T_2)$.
\end{definition}

\begin{definition}
Airy datum is said to be indecomposable if it is nontrivial and not isomorphic to a product of two nontrivial Airy data.
\end{definition}

\begin{proposition}
Every Airy datum is isomorphic to a product of finitely many indecomposable Airy data.
\end{proposition}
\begin{proof}
By induction on dimension.
\end{proof}


%



\section{Semisimple Lie algebras - general facts} \label{sec:semisimple_general}

\begin{definition}
Airy datum $(\mathfrak g, W, \Omega)$ is said to be homogeneous if $W$ is a linear representation of $\mathfrak g$ with purely quadratic moment map.
\end{definition}

\begin{proposition} \label{prop:homog_inv}
Let $(\mathfrak g, W,  \Omega)$ be a homogeneous Airy datum. Suppose that $\phi $ is~a~continuous function $W \to \mathbb C$ invariant under the $G$-action. Then $\phi(\Omega) = \phi(0)$.
\end{proposition}
\begin{proof}
For any $\lambda \in \mathbb C^{\times}$ we have $\lambda \Omega \in \Sigma_s$. In particular $\lambda \Omega$ belongs to the path component of $\Sigma_s$ containing $\Omega$. Since this set coincides with the orbit of $\Omega$, we have $\phi(\Omega)= \phi(\lambda \Omega)$, so $\phi (\Omega) = \lim \limits_{\lambda \to 0} \phi(\lambda \Omega) = \phi(0)$.
\end{proof}

Recall that the null cone $\mathcal N(W)$ of a representation $W$ of a group $G$ is defined as~the common zero locus of all $G$-invariant polynomials on $W$ homogeneous of~positive degree. It~is a basic object of interest in the classical invariant theory. Proposition \ref{prop:homog_inv} implies that for any homogeneous Airy datum $(\mathfrak g, W, \Omega)$ we have an~inclusion $\Sigma_s \subseteq \mathcal N(W)$. There exist classes of representations for which the~structure of the null cone is well understood. This makes Proposition \ref{prop:homog_inv} useful in~constraining homogeneous Airy data. We will now show that assumption of~homogeneity leads to no loss of generality in the case of semisimple Lie algebras.

\begin{proposition} \label{prop:Airy_homog}
Every Airy datum $(\mathfrak g, W, \Omega)$ with semisimple $\mathfrak g$ is isomorphic to~a~homogeneous Airy datum.
\end{proposition}
\begin{proof}
This fact was established in \cite{ABCD}. For completness we give two other proofs.

Choose a basis in $\mathfrak g$ and a symplectic affine coordinate system in $W$ centered at~$\Omega$ such that the moment map is represented by polynomials of the form (\ref{eq:cl_hamiltonians}). Decompose $\ell_i = y_i + Q_i$. Then $Q_i$ are homogeneous of degree two and satisfy $\{ Q_i, Q_j \} = f_{ij}^k Q_k$, so they furnish a linear representation of $\mathfrak g$ on the linear span of~$x^i$ and $y_i$. Furthermore relations (\ref{eq:cl_commutator}) imply the cocycle condition
\begin{equation}
\{ Q_i, y_j \} - \{ Q_j, y_i \} = f_{ij}^k y_k.
\end{equation}
By the Whitehead's lemma (discussed in the appendix \ref{app:cohom}), there exist coefficients $a_j$, $b^j$ such that
\begin{equation}
y_i = \{ Q_i, a_j x^j + b^j y_j \}.
\end{equation}
Now consider the affine automorphism of $W$ given by
\begin{equation}
x^i \mapsto x^i + b^i, \quad \quad y_i \mapsto y_i - a_i.
\end{equation}
Generators $\ell_i$ are mapped to $Q_i + \epsilon_i$, where $\epsilon_i$ are some constants. Relations (\ref{eq:cl_commutator}) then imply that $f_{ij}^k\epsilon_k =0$, so $\epsilon=0$. We have found a new affine coordinate chart in~which $\ell_i$ are purely quadratic, so the generated $G$-action is linear.

One may also avoid the use of Lie algebra cohomology\footnote{In fact this argument may be used to give a simple proof of Whitehead's lemma. Admittedly, it relies strongly on structure theory of semisimple Lie groups.}. Instead we choose a~maximal compact subgroup $K \subseteq G$. Using averaging techniques we may find a~fixed point of the action of $K$ on $W$. This fixed point is then also a fixed point for the action of whole $G$, by holomorphicity of the $G$-action. Expressing the moment map in coordinates centered at the fixed point we get vanishing linear term. Then commutation relations imply that the constant term also vanishes.
\end{proof}

From now on, we restrict attention to homogeneous Airy data $(\mathfrak g, W,  \Omega)$ whenever $\mathfrak g$ is semisimple. Our next step is to briefly review facts about symplectic representations of semisimple Lie algebras essential for further discussion.

Fix a semisimple Lie algebra $\mathfrak g$. Choose a Cartan subalgebra $\mathfrak h \subseteq \mathfrak g$ and a set of positive roots $\Delta_+$. We let $\Lambda \subseteq \mathfrak h^*$ be the lattice of integral weights, $C \subseteq \Lambda \otimes_{\mathbb Z} \mathbb R \subseteq \mathfrak h^*$ the (closed) fundamental Weyl chamber and $\Lambda_+ = \Lambda \cap C$ the set of dominant integral weights. We shall also consider the dual lattice $\Lambda^* = \{ H \in \mathfrak h | \forall \mu \in \Lambda \ \mu(H) \in \mathbb Z \} \subseteq \mathfrak h$. In other words, $\Lambda^*$ consists of these elements of $\mathfrak h$ which have integral eigenvalues in all finite-dimensional representation of $\mathfrak g$. For any $\lambda \in \Lambda_+$, denote the highest weight module with highest weight $\lambda$ by $V_{\lambda}$. Each $\lambda \in \Lambda_+$ has one of the following mutually exclusive properties:
\begin{itemize}
\item $V_{\lambda}$ is of real type, i.e. $H^0(\mathfrak g, \mathrm{Sym}^2 V_{\lambda}) \neq 0$,
\item $V_{\lambda}$ is of quaternionic type, i.e. $H^0(\mathfrak g, \Lambda^2 V_{\lambda}) \neq 0$,
\item $V_{\lambda}$ is of complex type, i.e. $V_{\lambda}^* \cong V_{\lambda^*}$ for some (unique) $\lambda^* \in \Lambda \setminus \{ \lambda \}$.
\end{itemize}
Any finite-dimensional representation $W$ of $\mathfrak g$ decomposes as $W = \bigoplus_{\lambda \in \Lambda} V_{\lambda}^{\oplus m_{\lambda}}$, with $m_{\lambda} \in \mathbb N$ vanishing for all but finitely many $\lambda$. It follows from the Schur's lemma that $W$ admits an invariant symplectic form if and only if $m_{\lambda}$ is even for $V_{\lambda}$ of real type and $m_{\lambda} = m_{\lambda^*}$ for $V_{\lambda}$ of complex type. In this situation the symplectic structure on $W$ is unique up to a $\mathfrak g$-module isomorphism. The $\mathfrak g$-action is automatically hamiltonian, with the moment map uniquely determined as $ \ell(q) ( T ) = \frac{1}{2} \omega(T q, q)$ for $q \in W$ and $T \in \mathfrak g$. Vector space $W$ decomposes as a direct sum of its weight spaces, $W = \oplus_{\mu \in \Lambda} W_{\mu}$ with 
\[
W_{\mu} = \{ w \in W | \forall H \in \mathfrak h \ H w = \mu(H) w \}
\] 
Subspace $W_{\mu}$ is~orthogonal to $W_{\nu}$ unless $\mu + \nu =0$. Element $\mu \in \Lambda$ is said to be a~weight of $W$ if~$W_{\mu} \neq 0$. Dimension of $W_{\mu}$ is called the multiplicity of $\mu$ in $W$.

\begin{proposition}
If $\mathfrak g$ is a semisimple Lie algebra, then there are finitely many isomorphism classes of Airy data of the form $(\mathfrak g, W,  \Omega)$.
\end{proposition}
\begin{proof}
The number of isomorphism classes of $W$ is finite, symplectic form $\omega$ is unique up to isomorphism and the moment map $\ell$ is uniquely determined by $W$ and $\omega$. Once $\mathfrak g$, $W$, $\omega$ and $\ell$ are fixed, space $\Sigma_s$ has finitely many connected components.
\end{proof}



We shall say that a symplectic $\mathfrak g$-module $W$ is admissible if there exists an Airy datum of the form $(\mathfrak g, W, \Omega)$. This is true if and only if the corresponding set $\Sigma_s$ is~nonempty. It turns out that many symplectic $\mathfrak g$-modules of dimension $2 \dim \mathfrak g$ are not admissible. To rule them out, we will need to better understand properties of~the~element $\Omega$. The~first steps in this direction are the following statements:

\begin{proposition} \label{prop:triv_adj_admissibility}
Let $(\mathfrak g, W,  \Omega)$ be an Airy datum with $\mathfrak g$ semisimple. Then
\begin{enumerate}
\item $H^0(\mathfrak g,W)=0$,
\item W is not isomorphic to $\mathfrak g \oplus \mathfrak g$.
\end{enumerate}
\end{proposition}
\begin{proof}
$1.$ $H^0(\mathfrak g, W)$ is a symplectic subspace of $W$. Thus if $H^0(\mathfrak g, W) \neq 0$, then $T_{\Omega} \Sigma_s \subseteq H^0(\mathfrak g, W)^{\perp} \subsetneq W$, so $T_{\Omega} \Sigma_s$ can't be Lagrangian. \\
$2.$ Suppose that the contrary is true. We write $\Omega = (\Omega_1, \Omega_2) \in \mathfrak g \oplus \mathfrak g$. Let $p : \mathfrak g \times \mathfrak g \to \mathfrak g$ be a Lie polynomial\footnote{This means that $p(T,S)$ is a linear combination of $T$, $S$, $[T,S]$, $[S,[T,S]]$ etc.} and $k>0$ a natural number and consider the function
\begin{equation}
\phi : \mathfrak g \times \mathfrak g \ni (T,S) \mapsto \mathrm{tr}_{\mathfrak g} \mathrm{ad}_{p(T,S)}^k \in \mathbb C.
\end{equation}
$\phi$ is continuous, $\mathfrak g$-invariant and $\phi(0)=0$. Therefore $\phi(\Omega)=0$ by Proposition~\ref{prop:homog_inv}. Since $p$ and $k$ were arbitrary, we conclude that for any element $T$ of the Lie subalgebra $\mathfrak n \subseteq \mathfrak g$ generated by $\Omega_1, \Omega_2$, traces of all powers of $\mathrm{ad}_T$ vanish. Thus $\mathrm{ad}_T$ is a nilpotent endomorphism of $\mathfrak g$, and hence also of the invariant subspace $\mathfrak n \subseteq \mathfrak g$. Since $T \in \mathfrak n$ was arbitrary, $\mathfrak n$ is a nilpotent Lie algebra. In particular its center $Z(\mathfrak n)$ is nontrivial. Let $T \in Z(\mathfrak n) \setminus \{ 0 \}$. Then $T \Omega =0$, a contradiction.
\end{proof}

\begin{proposition} \label{prop:J_exists}
Let $(\mathfrak g, W,  \Omega)$ be a homogeneous Airy datum. There exists a unique $J \in \mathfrak g$ such that $J \Omega=\Omega$. If $\mathfrak g$ is semisimple, $J$ is a semisimple element of $\mathfrak g$.
\end{proposition}
\begin{proof}
Since $\Sigma_s$ is a cone, $\Omega \in T_{\Omega} \Sigma_s = \{ T \Omega | T \in \mathfrak g  \}$, where we identified $T_{\Omega} W$ with $W$ itself. This proves the existence of $J$. If $J' \in \mathfrak g$ satisfies $J' \Omega$, then $J'=J$ (since the annihilator of $\Omega$ in $\mathfrak g$ is trivial). Now assume that $\mathfrak g$ is semisimple and let $J=J_{ss}+J_n$ be the Jordan-Chevalley decomposition of $J$. Then $J_{ss} \Omega = \Omega$ and~$J_n \Omega =0$, so $J_{ss} = J$.
\end{proof}

Recall \cite{Serre} that every semisimple element of a semisimple Lie algebra $\mathfrak g$ belongs to~some (not necessarily unique) Cartan subalgebra of $\mathfrak g$. Moreover action of~the~group of inner automorphisms of $\mathfrak g$ on the set of Cartan subalgebras is~transitive. Therefore we may fix a Cartan subalgebra $\mathfrak h \subseteq \mathfrak g$. Every Airy datum $(\mathfrak g, W,  \Omega)$ is~isomorphic to~one such that $J \in \mathfrak h$. From now on, we restrict attention to Airy data of this form. The next step is to further constrain the element $J$.

\begin{proposition} \label{prop:J_dimension}
Let $(\mathfrak g, W,  \Omega)$ be an Airy datum with $\mathfrak g$ semisimple. Consider the hyperplane $\mathcal H = \{ \mu \in \mathfrak h^* |  \mu(J)=1 \} \subseteq \mathfrak h^*$
and its subset $\Xi = \{ \mu \in \mathcal H | W_{\mu} \neq 0 \}$. Then
\begin{enumerate}
\item $0$ is not an element of $\mathcal H$.
\item $\Xi$ contains a basis of $\mathfrak h^*$.
\item Each of the triple $(J, \mathcal H, \Xi)$ uniquely determines the other two.
\item $J$ is rational, in the sense that $J \in \Lambda^* \otimes_{\mathbb Z} \mathbb Q$.
\item There exists an isomorphic Airy datum such that $\alpha(J) \geq 0$ for every $\alpha \in \Delta_+$.
\item For any root $\alpha$ there exists $\mu \in \Xi$ such that $\mu + \alpha$ is a weight of $W$, $W_{\mu + \alpha} \neq 0$.
\end{enumerate}
\end{proposition}
\begin{proof}
$1.$ Obvious. \\
$2.$ Suppose otherwise. Then there exists a nonzero $H \in \mathfrak h$ such that $\mu(H)=0$ for~each~$\mu \in \Xi$, so $H \Omega =0$. Contradiction. \\
$3.$ By construction, $J$ determines $\mathcal H$ and $\Xi$. $J$ is the unique element $T \in \mathfrak h$ such that $\mu(T)=1$ for every $\mu \in \mathcal H$. Since every basis of $\mathfrak h$ is contained in a unique hyperplane, one may reconstruct $\mathcal H$ from $\Xi$ as the unique hyperplane containing $\Xi$. \\
$4.$ Let $\Xi = \{ \mu_1, ..., \mu_m \}$. $J$ is uniquely determined by the affine system of equations $\mu_i(J)=1$. Since $\mu_i$ belong to $\Lambda \subseteq \Lambda \otimes_{\mathbb Z} \mathbb Q = (\Lambda^* \otimes_{\mathbb Z} \mathbb Q)^*$, this system has a solution in $\Lambda^* \otimes_{\mathbb Z} \mathbb Q \subseteq \mathfrak h$. Since solution of this system considered in $\mathfrak h$ is unique, $J \in \Lambda^* \otimes_{\mathbb Z} \mathbb Q$. \\
$5.$ By the previous point, $J$ must belong to the dual cone of some Weyl chamber. Since the Weyl group acts transitively on the set of Weyl chambers, we may assume that $J$ lies in the dual cone of the fundamental Weyl chamber. \\
$6.$ Assume otherwise. Then $ \mathfrak g_{\alpha}$ annihilates $\Omega$. Contradiction.
\end{proof}

It is natural to ask if point $4$ of the above Proposition can be strengthened, i.e.~if~$J$ belongs to the lattice $\Lambda^*$. One of the examples constructed in the Subsection \ref{sec:sl2_clas} shows that this is not necessarily true even if $\mathfrak g$ is simple. This leads to the~concept of the denominator of $J$, which is defined as the smallest positive integer $\mathrm{denom}(J)$ such that $\mathrm{denom}(J) \cdot J \in \Lambda^*$. Similarly for the point $5$, one can ask if condition $\alpha(J) \geq 0$ can be replaced by a strict inequality. This happens to be true for all simple Lie algebras, but there exist Airy structures for semisimple Lie algebras for which $J$ is~orthogonal to some root of $\mathfrak g$, i.e. such that $J$ is not a regular element of $\mathfrak g$. For~the benefit of the reader we recall the definition and properties of~regular elements of a semisimple Lie algebra in the Appendix \ref{app:ss_regular}.



For fixed $\mathfrak g$ and $W$ the number of weights of $W$ is finite, so points $2$ and $3$ of~Proposition \ref{prop:J_dimension} determine $J$ up to a finite ambiguity. This ambiguity is reduced by~imposing the additional condition $\alpha(J) \geq 0$ for every root $\alpha$. Many of the remaining candidates for $J$ may be excluded by the following fact.

\begin{proposition} \label{prop:J_spectrum}
Let $(\mathfrak g, W, \Omega)$ be an Airy datum with $\mathfrak g$ semisimple and let $  \lambda_1, ..., \lambda_n$ be the eigenvalues of $\mathrm{ad}_J$. Then
\begin{enumerate}
\item Each $\lambda_i$ is a rational number.
\item Multiplicity of any $\lambda$ among $\lambda_1, ..., \lambda_n$ is equal to the multiplicity of $- \lambda$. \\
 In particular $\sum_{i=1}^n \lambda_i =0$.
\item Spectrum of $J$ acting in $W$ takes the form
\begin{equation}
\mathrm{spec}_W(J) = \{ 1+ \lambda_1, ..., 1 + \lambda_n, -1 - \lambda_1, ..., -1 - \lambda_n \}.
\label{eq:J_spectrum_formula}
\end{equation}
\end{enumerate}
\end{proposition}
\begin{proof}
$1.$ Special case of $4.$ in Proposition \ref{prop:J_dimension}. \\
$2.$ Follows from the fact that $J$ is an infinitesimal symplectomorphism. \\
$3.$ Let $\{ T_i \}_{i=1}^n$ be a basis of $\mathfrak g$ with $[J,T_i] = \lambda_i T_i$. Put $e_i = T_i \Omega$. Vectors $e_i$ span a~Lagrangian subspace $T_{\Omega} \Sigma_s \subseteq W$ and satisfy $J e_i = (1 + \lambda_i) e_i$. To complete the proof, it is sufficient to show that there exists a Lagrangian complement $V$ of $T_{\Omega} \Sigma_s$ spanned by vectors $\{ f_i \}_{i=1}^n$ with $J f_i = - (1 + \lambda_i) f_i$. We proceed inductively. First notice that $\mathrm{ker}(J-1 - \lambda_1)+ \mathrm{ker}(J+1+\lambda_1)$ is a symplectic subspace of $W$, so~we can~find $f_1$ with $\omega(e_i,f_1) = \delta_{i1}$ and $J f_1 = - (1+ \lambda_1) f_1$. Now suppose that we~have found $\{ f_1,...,f_k \}$ for some $1 \leq k < n$. Applying the same argument to the orthogonal complement
of~the~symplectic subspace spanned by $\{ e_i, f_i \}_{i=1}^k$ we find $f_{k+1}$.
\end{proof}

It is of interest to classify indecomposable Airy data for semisimple Lie~algebras. This doesn't reduce to classification of Airy data for simple Lie~algebras. Indeed, explicit examples of indecomposable Airy data for semisimple Lie algebras which are not simple are presented in Section \ref{sec:ss_examples}. Here we derive a simple criterion for~indecomposability and prove uniqueness of indecomposable factors.

\begin{definition}
Let $(\mathfrak g, W,  \Omega)$ be an Airy datum with $\mathfrak g$ semisimple. We define its associated graph by taking the simple factors of $\mathfrak g$ as vertices, with an edge between two simple factors $\mathfrak g'$ and $\mathfrak g''$ if and only if $W$ contains an irreducible submodule on~which both $\mathfrak g'$ and $\mathfrak g''$ act nontrivially.
\end{definition}

\begin{proposition}
Airy datum $(\mathfrak g, W,  \Omega)$ with $\mathfrak g$ semisimple is indecomposable if and only if its associated graph is connected.
\end{proposition}
\begin{proof}
Clearly $(\mathfrak g, W, \Omega)$ is indecomposable if its associated graph $G$ is connected. Now suppose that $G$ is not connected. Then we may decompose $\mathfrak g = \mathfrak g_1 \times \mathfrak g_2$ (with~both factors nonzero), $W= W_1 \oplus W_2 $. In this situation $\Sigma_s$ is the product of~the~corresponding sets for $(\mathfrak g_1, W_1)$ and $(\mathfrak g_2, W_2)$, so also $\Omega$ factorizes.
\end{proof}

We remark that formation of the associated graph is a contravariant functor from the category of Airy structures for semisimple Lie algebras to the category of graphs. 

\begin{proposition}
Let $\{ (\mathfrak g_i, W_i, \Omega_i) \}_{i=1}^n$ and $\{ (\mathfrak g_i',W_i',\Omega_i')\}_{i=1}^m$ be indecomposable Airy data with each $\mathfrak g_i$ and $\mathfrak g_i'$ semisimple. Suppose that
\begin{equation}
(\phi ,f) : \prod_{i=1}^n (\mathfrak g_i, W_i, \Omega_i) \to \prod_{i=1}^m (\mathfrak g_i', W_i', \Omega_i')
\end{equation}
is an isomorphism. Then $m = n$ and (possibly after a~permutation) there exist isomorphisms $(\phi_i, f_i) : (\mathfrak g_{i},W_{i}, \Omega_{i}) \to (\mathfrak g'_i, W'_i, \Omega'_i)$ such that $\phi = \prod_{i=1}^n \phi_i$, $f = \prod_{i=1}^n f_i$.
\end{proposition}
\begin{proof}
We identify factors of $\prod_{i=1}^n \mathfrak g_i$ with their images in $\prod_{i=1}^m \mathfrak g_i'$ through $\phi$. Using the fact that simple factors of a semisimple Lie algebra are uniquely determined and functoriality of the associated graph construction we see that (after a permutation) we have $m=n$ and $\mathfrak g_i' = \mathfrak g_i$. Then clearly $f = \prod_{i=1}^n f_i$ for some module isomorphisms $f_i : W_i \to W_i'$. By construction, $f_i(\Omega_i) = \Omega'_i$.
\end{proof}

We close this section with a remark that in all examples of Airy data $(\mathfrak g, W, \Omega)$ constructed in this paper $\Omega$ is a cyclic vector for $W$. We have not managed to decide if this is always true for $\mathfrak g$ semisimple. Below we prove a weaker statement.

\begin{proposition} \label{prop:Omega_almost_cyclic}
Let $(\mathfrak g, W, \Omega)$ be a nontrivial Airy datum with $\mathfrak g$ semisimple. Then the~submodule of $W$ generated by $\Omega$ has dimension strictly greater than $\dim \mathfrak g$.
\end{proposition}
\begin{proof}
Let $W' \subseteq W$ be the submodule generated by $\Omega$. Since $W'$ contains the Lagrangian subspace $T_{\Omega} \Sigma_s$, we have $\dim W' \geq \dim \mathfrak g$. Suppose that this inequality is saturated. Let $\{ \lambda_i \}_{i=1}^{\dim \mathfrak g}$ be the eigenvalues of $\mathrm{ad}_J$. Then the eigenvalues of $J$ acting in $W'$ are $\{ 1 + \lambda_i \}_{i=1}^{\dim \mathfrak g}$, which leads to an absurd chain of equalities
\begin{equation}
0 = \mathrm{tr}_{W'}(J) = \dim (\mathfrak g) + \mathrm{tr}_{\mathfrak g} (\mathrm{ad}_J)= \dim (\mathfrak g).
\end{equation}
\end{proof}

\section{Automorphisms of Airy data} \label{sec:automorphisms}

\begin{proposition} \label{prop:stab_finitness}
Suppose that $(\mathfrak g, W, \Omega)$ is an Airy datum with $\mathfrak g$ semisimple. Then $\mathrm{Stab}(\Omega)$ is a finite group. Moreover $\mathrm{Ad}_g(J)=J$ for any $g \in \mathrm{Stab}(\Omega)$.
\end{proposition}
\begin{proof}
$G$ is a linear algebraic group acting algebraically on $W$. Therefore $\mathrm{Stab}(\Omega)$ is Zariski closed. Furthermore we have $\dim(\mathrm{Stab}(\Omega))=0$, for otherwise there would exist an element of $\mathfrak g \setminus \{ 0 \}$ annihilating $\Omega$. Thus $\mathrm{Stab}(\Omega)$ is finite. Now~pick some $g \in \mathrm{Stab}(\Omega)$. We have $\Omega = g J g^{-1} \Omega$, so $\mathrm{Ad}_g(J) = J$ by uniqueness~of~$J$.
\end{proof}

We remark that Proposition \ref{prop:stab_finitness} is false if the assumption of semisimplicity of~$\mathfrak g$ is~dropped. In general $G$ does not come equipped with a~canonical structure of~an~algebraic variety. Even if such structure exists, it may happen that the~$G$-action on~$W$ is not algebraic. This is the case in some of the examples of~Airy data discussed in~\cite{Analyticity}, in which $\mathrm{Stab}(\Omega)$ was found to be infinite cyclic.

\begin{definition}
Let $(\phi,f)$ be an automorphism of an Airy datum $(\mathfrak g, W, \Omega)$. We~shall say that $(\phi,f)$ is inner (resp. almost inner) if $f=g$ (resp. $\phi = \mathrm{Ad}_g$) for some $g \in G$. Group of inner (resp. almost inner) automorphisms of $(\mathfrak g, W, \Omega)$ will be denoted by~$\mathrm{Inn}(\mathfrak g , W, \Omega)$ (resp. $\mathrm{AInn}(\mathfrak g, W, \Omega)$).
\end{definition}


\begin{proposition}
Let $(\phi,f)$ be an automorphism of Airy datum $(\mathfrak g, W, \Omega)$.
\begin{enumerate}
\item If $f=g$ for some $g \in G$, then $\phi = \mathrm{Ad}_g$. In particular every inner automorphism is almost inner.
\item We have $\mathrm{Inn}(\mathfrak g, W,\Omega) \cong \frac{\mathrm{Stab}(\Omega)}{\mathrm{Stab}(W)}$, where $\mathrm{Stab}(W) = \{ g \in G | \forall q \in W \ g \cdot q = q \}$. In particular if $\mathfrak g$ is semisimple, then $\mathrm{Inn}(\mathfrak g, W , \Omega)$ is a finite group.
\item $\mathrm{Inn}(\mathfrak g, W, \Omega)$ is a normal subgroup of $\mathrm{AInn}(\mathfrak g, W, \Omega)$.
\item $\mathrm{AInn}(\mathfrak g, W, \Omega)$ is a normal subgroup of $\mathrm{Aut}(\mathfrak g, W, \Omega)$.
\item Suppose that $(\mathfrak g, W, \Omega)$ is homogeneous. Then $\phi(J)=J$.
\item Suppose that $\mathfrak g$ is semisimple and $J$ is regular. Then $\mathrm{Stab}(\Omega)$ is contained in~the~subgroup $e^{\mathfrak h} \subseteq G$ generated by $\mathfrak h$. In particular $\mathrm{Stab}(\Omega)$ is abelian.
\end{enumerate}
\end{proposition}
\begin{proof}
$1.$ Pick $T \in \mathfrak g$. Formula (\ref{eq:Airy_morphism_intertwines}) gives $\xi(\phi (\mathrm{Ad}_g^{-1}(T))) = \xi(T)$, so by faithfulness of $\xi$ we have $ \phi(\mathrm{Ad}_g^{-1}(T)) =T$. \\
$2.$ Faithfulness of $\xi$ implies that $\mathrm{Stab}(W)$ is a discrete normal (and hence central) subgroup of $G$. Thus $\mathrm{Stab}(\Omega) \ni g \mapsto (\mathrm{Ad}_g, g) \in \mathrm{Inn}(\mathfrak g, W, \Omega)$ is an epimorphism with kernel $\mathrm{Stab}(W)$. \\
$3.$ Suppose that $(\phi,f) \in \mathrm{AInn}(\mathfrak g, W, \Omega)$. Pick $g \in G$ such that $\phi = \mathrm{Ad}_g$. Then $F=g^{-1}f $ commutes with the $G$-action on $W$, so for any $h \in \mathrm{Stab}(\Omega)$ we have
\begin{equation}
(\phi,f) \circ (\mathrm{Ad}_h,h) \circ (\phi ,f )^{-1} = (\mathrm{Ad}_{ghg^{-1}}, ghg^{-1}).
\end{equation}
$4.$ Group of inner automorphism of $\mathfrak g$ is a normal subgroup of $\mathrm{Aut}(\mathfrak g)$. \\
$5.$ We have $\phi(J) \Omega = \Omega$, so $\phi(J)=J$ by uniqueness of $J$. \\
$6.$ Pick some $g \in \mathrm{Stab}(\Omega)$. Then $\mathrm{Ad}_g(J) = J$, so $\mathrm{Ad}_g(\mathfrak h) = \mathfrak h$ by regularity of $J$. If~$g \notin e^{\mathfrak h}$, there exists a root $\alpha$ such that $\alpha(\mathrm{Ad}_g(J)) < 0$. Contradiction.
\end{proof}

Recall that real structure on a complex vector space $V$ is an antilinear involution $\sigma: V \to V$. The set $V^{\sigma} = \{ v \in V | \sigma(v)=v \}$ of fixed points of $\sigma$ is a real subspace of~$V$ with $V^{\sigma} \otimes_{\mathbb R} \mathbb C = V$. Conversely, given a real subspace $V' \subseteq V$ with $V' \otimes_{\mathbb R} \mathbb C$ there exists a unique real structure $\sigma$ on $V$ such that $V'=V^{\sigma}$. Now let $\mathfrak g$ be a~complex Lie algebra. Antilinear involution $\sigma$ on $\mathfrak g$ is said to be a real structure of $\mathfrak g$ if it is a~homomorphism of real Lie algebras. In this situation $\mathfrak g^{\sigma}$ is a real Lie algebra. If the Killing form on $\mathfrak g^{\sigma}$ is negative-definite, we say that $\sigma$ is a compact real form. In this situation $\mathfrak g$ is semisimple and $\mathfrak g^{\sigma}$ is the Lie algebra of a simply-connected compact Lie group $G^{\sigma}$. Let $W$ be a representation of $\mathfrak g$. Real structure $K$ on $W$ is said to be compatible with $\sigma$ if $K(Tq)=\sigma(T) K(q)$ for $T \in \mathfrak g$, $q \in W$, or equivalently if $W^K$ is~a~representation of $\mathfrak g^{\sigma}$ and $W = W^K \otimes_{\mathbb R} \mathbb C$ as a $\mathfrak g^{\sigma}$-module. In this situation we shall abuse the notation by denoting the involution $K$ simply by $\sigma$. We remark that real structures on affine representations of $\mathfrak g$ may also be defined, but by Proposition \ref{prop:Airy_homog} we shall not need them here.

\begin{definition}
Let $A=(\mathfrak g , W, \Omega)$ be a homogeneous Airy datum. A real structure on~$A$ is a real structure $\sigma$ on $\mathfrak g$ together with a compatible real structure $\sigma$ on $W$ such that $\sigma(\Omega)=\Omega$.
\end{definition}

\begin{proposition}
Let $\sigma$ be a real structure on a nontrivial homogeneous Airy datum $(\mathfrak g, W , \Omega)$. Then $\sigma$ is not compact.
\end{proposition}
\begin{proof}
We have $\sigma(J) \Omega = \Omega$, so $\sigma(J) = J$ by uniqueness of $J$. Since $J$ belongs to~$\mathfrak g^{\sigma}$, we have $\phi( \Omega) = \phi(0)$ for every continuous, $G^{\sigma}$-invariant function $\phi : W \to \mathbb C$. By~averaging techniques we may construct a $G^{\sigma}$ invariant norm $\| \cdot \|$ on $W$. It follows that $\| \Omega \| =0$, so $\Omega =0$. Then $\left. d \ell \right|_{\Omega} =0$, a contradiction.
\end{proof}

As illustrated by examples in Sections \ref{sec:simple_class} and \ref{sec:ss_examples}, noncompact real forms do exist, at least for some Airy data.

\section{Simple Lie algebras - classification} \label{sec:simple_class}

\begin{proposition} \label{prop:dim_bounds}
We list isomorphism classes of symplectic representations of~simple Lie algebras whose admissibility is not ruled out by the Proposition \ref{prop:triv_adj_admissibility}. Whenever $\mathfrak g$ is a~classical Lie algebra, we denote the tautological representation by~$F$. In the~case of symplectic algebras, we let $\Lambda^k_0 F$, $k \in \mathbb N$ be the subspace of these elements of $\Lambda^k F$ whose any contraction with the symplectic form of $F$ vanishes.
\begin{itemize}
\item $\mathfrak{sl}_2$ : $F^{\oplus 3}$, $F \oplus \mathrm{Sym}^3 F $, $\mathrm{Sym}^5 F$.
\item $\mathfrak{sl}_6$ : $\Lambda^2 F \oplus \Lambda^4 F \oplus (\Lambda^3 F)^{\oplus 2}$.
\item $\mathfrak{sp}_4$ : $F^{\oplus 5}$, $\widetilde F^{\oplus 4}$, $\mathrm{Sym}^3 F$, $F \otimes \widetilde F$. Here $\widetilde F = \Lambda_0^2 F$ is the tautological representation of $\mathfrak{so}_5$, which is isomorphic to $\mathfrak{sp}_4$.
\item $\mathfrak{sp}_6$ : $F^{\oplus 7}$, $(\Lambda_0^2 F)^{\oplus 2} \oplus \Lambda_0^3 F$, $(\Lambda_0^3 F)^{\oplus 3}$.
\item $\mathfrak{sp}_8$ : $F^{\oplus 9}$, $F^{\oplus 3} \oplus \Lambda_0^3 F$.
\item $\mathfrak{sp}_{10}$ : $F^{\oplus 11}$, $\Lambda_0^3 F$.
\item $\mathfrak{sp}_{2k}$, $k \geq 6$ : $F^{\oplus 2k+1}$.
\item $\mathfrak{so}_{2k+1}$, $k \geq 3$ : $F^{\oplus 2k}$.
\item $\mathfrak g_2$ : $F^{\oplus 4}$, where $F$ is the unique irreducible representation of dimension $7$.
\item $\mathfrak f_4$ : $F^{\oplus 4}$, where $F$ is the unique irreducible representation of dimension $26$.
\end{itemize}
\end{proposition}
\begin{proof}
First note \cite[p.~217-218]{BourbakiIII} that the only simple Lie algebras $\mathfrak g$ which admit an irreducible symplectic representation of dimension at most $2 \dim (\mathfrak g)$ are $\mathfrak{sl}_6, \mathfrak{so}_{11}$, $\mathfrak{so}_{12}$, $\mathfrak{so}_{13}$, $\mathfrak e_7$ and the symplectic Lie algebras. Furthermore for $n \geq 6$ the only irreducible symplectic representation of $\mathfrak{sp}_{2n}$ of desired dimension is the tautological representation. As for irreducible representations which are not symplectic, it~is~sufficient to consider those of dimension at most $\dim(\mathfrak g)$. Complete list of such representations is given in \cite[p.~414,~531-532]{FultonHarris}. Having established which representations may appear in the decomposition of $W$, one has to find all ways to add them together to get a representation of dimension $2 \dim (\mathfrak g)$. The end result of~this calculation is the table above.
\end{proof}

Our next goal is to determine which representations among those listed in~the~Proposition \ref{prop:dim_bounds} are admissible. The following fact rules out all but finitely many candidates.

\begin{proposition} \label{prop:exclude_fundamental}
Let $\mathfrak g= \mathfrak{sp}_{2k}$, $k \geq 1$ and $W = F^{\oplus 2k+1}$ or $\mathfrak g = \mathfrak{so}_{2k+1}$, $k \geq 1$ and $W = F^{\oplus 2k}$. Then $W$ is not admissible.
\end{proposition}
\begin{proof}
We present the proof for $\mathfrak g = \mathfrak{sp}_{2k}$. The second case is handled analogously. \\
Suppose that $(\mathfrak g, W, \Omega)$ is an Airy datum. Write $\Omega = (\Omega_1,...,\Omega_{2k+1})$, with $\Omega_i \in F$. Proposition \ref{prop:homog_inv} implies that elements $\Omega_i$ are pairwise orthogonal with respect to the symplectic form of $F$. Therefore they are contained in some Lagrangian subspace $L \subseteq F$. It~is easy to check that there exists a nonzero element $T \in \mathfrak g$ annihilating $L$. Thus~$T \Omega =0$, which is absurd.
\end{proof}

Most of the remaining representations are ruled out by the following construction. If~$\mathfrak g$ is simple, its invariant bilinear form is unique up to scale. Thus for any representation $W$ there is a real\footnote{One can show that this number is always positive and rational. More precisely, if $h^{\vee}(\mathfrak g)$ is the dual Coxeter number of $\mathfrak g$ (which is natural), then $h^{\vee}(\mathfrak g) \mathrm{ind}(W) \in \mathbb N$. We shall not use this result.} number $\mathrm{ind}(W)$ (called the index of $W$) such that $\mathrm{tr}_W(TS) = \mathrm{ind}(W) \mathrm{tr}_{\mathfrak g}(\mathrm{ad}_T \mathrm{ad}_S)$ for any $T, S \in \mathfrak g$.

\begin{proposition} \label{prop:trace_formula}
Let $(\mathfrak g, W,  \Omega)$ be an Airy datum with $\mathfrak g$ simple. Then
\begin{equation}
(\mathrm{ind}(W)-2) \mathrm{tr}_{\mathfrak g} ( \mathrm{ad}_J^2) = 2 \dim(\mathfrak g).
\label{eq:trace_formula}
\end{equation}
In particular we have an estimate
\begin{equation}
\mathrm{ind}(W) > 2.
\end{equation}
\end{proposition}
\begin{proof}
Let $\lambda_1, ..., \lambda_n$ be the eigenvalues of $\mathrm{ad}_J$. By Proposition \ref{prop:J_spectrum} we have
\begin{equation}
\mathrm{ind}(W) \mathrm{tr}_{\mathfrak g} ( \mathrm{ad}_J^2) = \mathrm{tr}_W (J^2) = 2 \sum_{i=1}^n (1+\lambda_i)^2 = 2 \mathrm{tr}_{\mathfrak g} ( \mathrm{ad}_J^2) + 2n,
\end{equation}
where we used $\sum_{i=1}^n \lambda_i =0$. Rearrangement of this equation yields (\ref{eq:trace_formula}). \\
Since the eigenvalues of $\mathrm{ad}_J$ are rational and not all equal to zero, $\mathrm{tr}_{\mathfrak g} ( \mathrm{ad}_J^2) >0$. Similarly, $\dim(\mathfrak g) >0$. Therefore equation (\ref{eq:trace_formula}) enforces that $\mathrm{ind}(W)-2 >0$.
\end{proof}

Computation of indices of representations listed in Proposition \ref{prop:dim_bounds} excludes all simple Lie algebras except of $\mathfrak{sl}_2$, $\mathfrak{sp}_4$ and $\mathfrak{sp}_{10}$. Each of these algebras admits two non-isomorphic Airy data, as we will demonstrate by explicit calculations.


\subsection{Lie algebra $\mathfrak{sl}_2$} \label{sec:sl2_clas}

Due to the isomorphism $\mathfrak{sl}_2 \cong \mathfrak{sp}_2$, admissibility of the representation $F^{\oplus 3}$ is~excluded by the Proposition \ref{prop:exclude_fundamental}. We will show that $F \oplus \mathrm{Sym}^3 F$ and $\mathrm{Sym}^5 F$ are~admissible, and that there exist two isomorphism classes of Airy data for $\mathfrak{sl}_2$.

Let $H, X, Y$ be the standard basis \cite{FultonHarris} of $\mathfrak{sl}_2$. These elements satisfy
\begin{equation}
[H,X] =  2X, \quad \quad [H,Y] = -2 Y, \quad \quad [X,Y]=H.
\end{equation}
We work with the canonical basis of $F$, $e_1 =\begin{pmatrix} 1 \\0 \end{pmatrix}$, $e_2 = \begin{pmatrix} 0 \\ 1 \end{pmatrix}$. Symplectic form on~$F$ is defined by $\omega(e_1,e_2)=1$, with remaining matrix elements fixed by bilinearity and skew-symmetry. Define
\begin{equation}
e_{i_1...i_k} = \mathrm{Sym}^k \left( e_{i_1} \otimes ... \otimes e_{i_k} \right) \in \mathrm{Sym}^k F.
\end{equation}
Set $\{ e_{i_1... i_k} \}_{1 \leq i_1 \leq ... \leq i_k \leq 2}$ is a basis of $\mathrm{Sym}^k F$. This module is symplectic if $k$ is odd, with the symplectic form determined by the equation
\begin{equation}
\omega \left( e_{i_1 ... i_k}, e_{j_1... j_k} \right) = \sum_{\sigma \in S_k} \prod_{l=1}^k \omega \left( e_{i_l},e_{j_{\sigma(l)}} \right) .
\end{equation}

Consider first the representation $W = F \oplus \mathrm{Sym}^3 F$. Spectrum of $H$ in $W$ is~$\{ \pm 3, \pm 1, \pm 1 \}$, so the only candidate for $J$ is $H$. Projection of $\Omega$ onto each of the summand of $W$ must be nonzero (for otherwise the linear span of $H \Omega$, $X \Omega$ and $Y \Omega$ could not be Lagrangian), so~we have $\Omega = ( s e_1, t e_{112})$ with some $s , t \in \mathbb C^{\times}$. Acting with a diagonal element of $\mathrm{SL}_2$ we may put $s=1$. Simple calculation shows that then assumptions of the Proposition \ref{prop:Omega_criterion} are satisfied if and only if $4 t^2 = 1$. We define
\begin{equation}
\Omega_{\pm} = \left( e_1, \pm \frac{1}{2} e_{112} \right).
\end{equation}
By construction, $\Omega_{\pm} \in \Sigma_s$. It's easy to check that $\mathrm{Stab}(\Omega_{\pm})=0$ and that there exists no element $g \in \mathrm{SL}_2$ such that $g \cdot \Omega_+ = \Omega_-$. Therefore we conclude that
\begin{equation}
\Sigma_s \cong \mathrm{SL}_2 \sqcup \mathrm{SL}_2.
\end{equation}
Even though $\Sigma_s$ is disconnected, Airy data corresponding to distinct connected components are still isomorphic. Indeed, the two connected components of $\Sigma_s$ are interchanged by the $\mathfrak g$-module automorphism $W \ni (u,v) \mapsto (u,-v) \in W$.

Case $W = \mathrm{Sym}^5 F$ is handled similarly, with the result that one can take $J = \frac{H}{3}$, $\Omega = e_{11112}$. Space $\Sigma_s$ is connected, but in this case the stabilizer of $\Omega$ is nontrivial:
\begin{equation}
\mathrm{Stab}(\Omega) = \left \{ 1 , \exp \left( \pm \frac{2\pi i}{3} H \right) \right \} \subseteq \mathrm{SL}_2.
\end{equation}
In contrast to the previous example, $\mathrm{Stab}(\Omega)$ is not a normal subgroup of $\mathrm{SL}_2$. Thus $\mathrm{Orb}(\Omega)$ is not a Lie group. Nevertheless, $W$ is admissible and we have $\Sigma_s \cong \frac{\mathrm{SL}_2}{\mathbb Z_3}$. We~remark that this Airy datum was constructed for the first time in \cite{ABCD}.

We remark that $\mathrm{Aut}(\mathfrak g, W, \Omega) = \mathrm{Inn}(\mathfrak g, W, \Omega)$ for Airy data constructed in this section. This happens to be true for all Airy data for simple Lie algebras. For~semisimple Lie algebras both $\frac{\mathrm{Aut}(\mathfrak g, W, \Omega)}{\mathrm{AInn}(\mathfrak g, W, \Omega)}$ and $\frac{\mathrm{AInn}(\mathfrak g, W, \Omega)}{\mathrm{Inn}(\mathfrak g, W, \Omega)}$ may be nontrivial, as~demonstrated by examples in Section \ref{sec:ss_examples}.

Airy data admit a real structure $\sigma$ with $\mathfrak g^{\sigma} = \mathfrak{sl}_2(\mathbb R)$. It is defined by $\sigma(Z)=Z$ for $Z \in \{ H,X,Y \}$, $\sigma(e_i)=e_i$ for $i \in \{1,2 \}$ and extended to other representations by~demanding that $\sigma$ is a homomorphism of the tensor algebra.

\subsection{Lie algebra $\mathfrak{sp}_4$} \label{sec:sp4_Airy}

We will now consider the Lie algebra $\mathfrak g = \mathfrak{sp}_4$. Representations $F^{\oplus 5}$ and $\widetilde F^{\oplus 4}$ are ruled out by Proposition \ref{prop:exclude_fundamental}. We will show that $\mathrm{Sym}^3 F$ is also not admissible, while $F \otimes \widetilde F$ admits two non-isomorphic Airy data.



\begin{figure}[h]
\centering
\includegraphics[width=13cm]{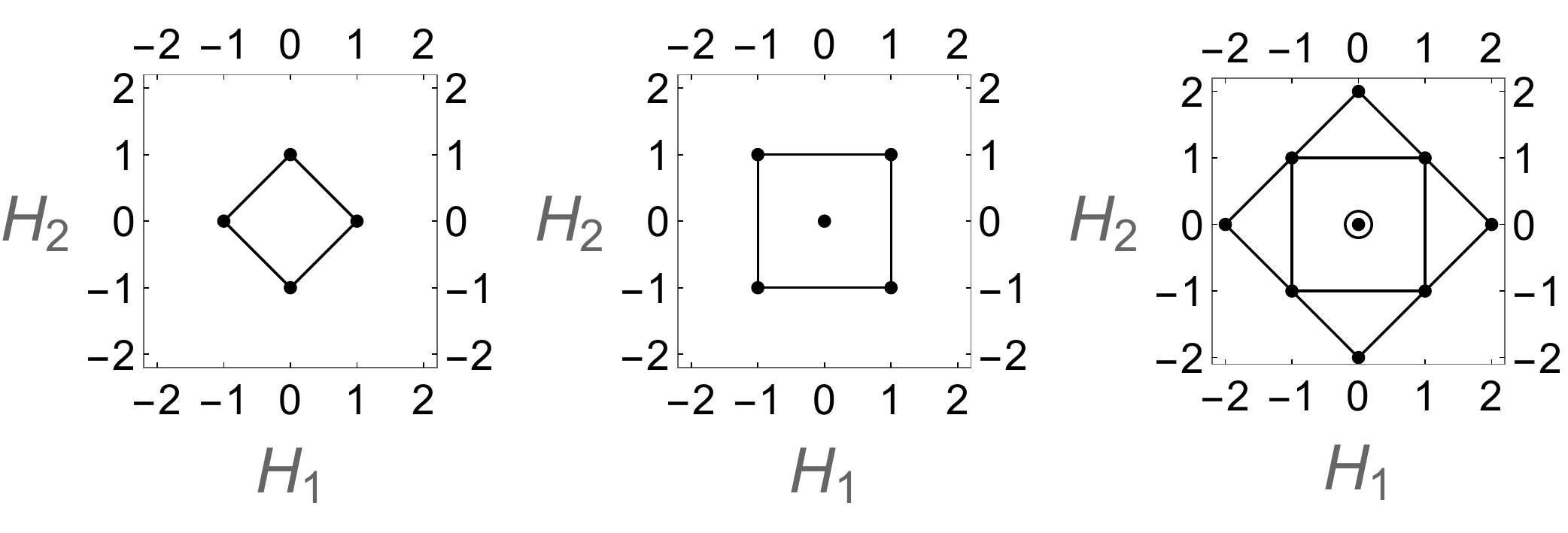}
\caption{Weight diagrams for representations $F, \widetilde F$ and the adjoint of $\mathfrak{sp}_4$. Weights are represented by dots, with surrounding circles indicating multiplicities. We draw parallelograms invariant under the Weyl group action to help the reader to see the~symmetry of the diagrams.}
\label{fig:F_Ft_Adj_Weights}
\end{figure}

We choose the standard \cite{FultonHarris} Cartan subalgebra, set of positive roots and basis
\begin{equation}
\{ H_1, H_2, U_1, U_2, V_1, V_2, X_{12}, X_{21}, Y_{12}, Z_{12} \}
\end{equation}
in $\mathfrak g$. Tautological representation is spanned by $e_1, e_2, e_3, e_4$, with symplectic form whose only (up to skew-symmetry) nonzero matrix elements are
\begin{equation}
\omega(e_1,e_3) = \omega(e_2, e_4) = 1.
\end{equation}

Representation $\widetilde F$ is a codimension one direct summand in $\Lambda^2 F$. Thus we put $e_{ij} = e_i \wedge e_j$ for $1 \leq i < j \leq 4$. Scalar product on $\Lambda^2 F$ is defined by
\begin{equation}
(e_{ij}, e_{kl}) = 2 \omega (e_i, e_k) \omega(e_j, e_l) - 2 \omega(e_i, e_l) \omega(e_j, e_k).
\end{equation}
We define also $\eta = e_{13} - e_{24}$. Set $\{ e_{12}, e_{23}, e_{34}, e_{14}, \eta \}$ is a basis of $\widetilde F$. Finally, the~symplectic form on $F \otimes \widetilde F$ is defined first on decomposable tensors
\begin{equation}
\omega(x_1, \otimes y_1, x_2 \otimes y_2) = \omega(x_1,x_2) (y_1, y_2) \quad \quad \mathrm{for} \ x_1, x_2 \in F, y_1, y_2 \in \widetilde F
\end{equation}
and extended to the whole space by bilinearity. Weight diagrams for the most basic representations of $\mathfrak g$ are presented in Figure \ref{fig:F_Ft_Adj_Weights}. We shall also consider slightly more complicated representations $\mathrm{Sym}^3 F$ and $F \otimes \widetilde F$. It will be important that the latter is reducible. More precisely, contraction with the symplectic form yields a~nonzero $\mathfrak g$-module epimorphism $\mathrm{tr} : F \otimes \widetilde F \to F$. Kernel of this map is~an~irreducible representation, which we denote by $F^{\perp}$.

\begin{figure}[h]
\centering
\includegraphics[width=12cm]{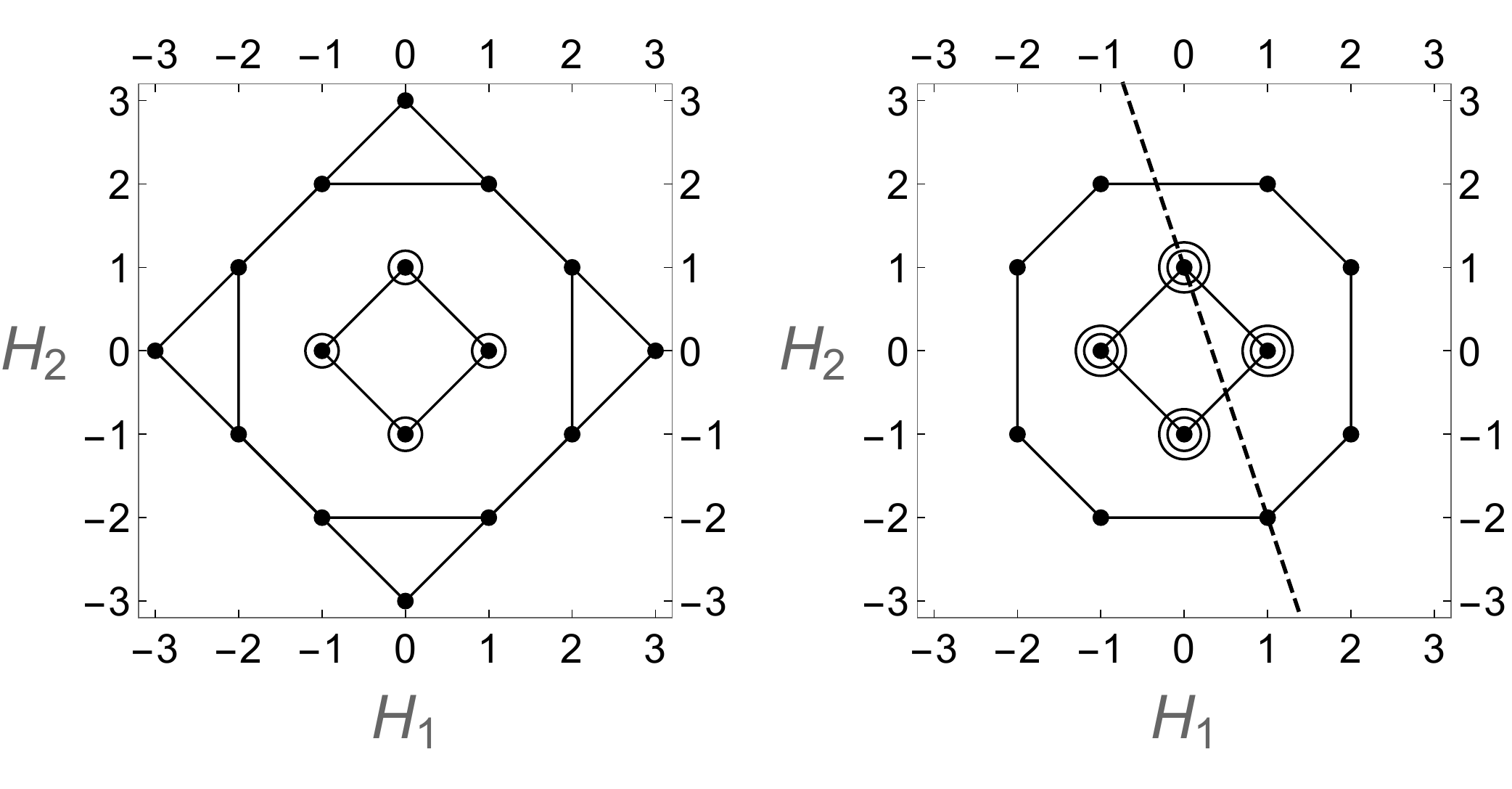}
\caption{Weight diagrams for representations $\mathrm{Sym}^3 F$ and $F \otimes \widetilde F$ of $\mathfrak{sp}_4$. Dashed line represents the set $\mathcal H$ for Airy structures found in this section.}
\label{fig:Sym3F_W_Weights}
\end{figure}

Examination of the weight diagrams of the adjoint representation and of $\mathrm{Sym}^3 F$ (see Figure \ref{fig:Sym3F_W_Weights}) shows that the only possible forms of $J$ not excluded by the Proposition \ref{prop:J_dimension} are $H_1, 3H_1 + H_2$ and $\frac{3H_1 + H_2}{3}$. Proposition \ref{prop:trace_formula} yields $\mathrm{tr}_{\mathfrak g} (\mathrm{ad}_J^2) = \frac{40}{3}$, which is not true for any of the candidates. Thus $\mathrm{Sym}^3 F$ is not admissible. In the case of $F \otimes \widetilde F$, the only candidates for $J$ are $3H_1 + H_2$ and $H_1 + H_2$. Proposition \ref{prop:trace_formula} gives $\mathrm{tr}_{\mathfrak g}(\mathrm{ad}_J^2)=120$. This is satisfied for $3H_1 + H_2$. Spectral test is also passed:
\begin{subequations}
\begin{gather}
\mathrm{spec}_{\mathfrak g} ( \mathrm{ad}_J ) = \{ 0,0, \pm 2, \pm 2, \pm 4, \pm 6 \}, \\
 \mathrm{spec}_W(J) = \{ \pm 1, \pm 1, \pm 1 \pm 1, \pm 3, \pm 3, \pm 3, \pm 5, \pm 5, \pm 7 \}.
 \end{gather}
\end{subequations}

We put $W= F \otimes \widetilde F$, $J=3H_1 + H_2$ and look for $\Omega \in \Sigma_s \subseteq W$ satisfying $J \Omega = \Omega$.  General solution of this eigenvalue equation takes the form
\begin{equation}
\Omega = s e_1 \otimes e_{23} + t e_3 \otimes e_{12} + u e_2 \otimes \eta + v e_4 \otimes e_{14}
\end{equation}
with some $s,t, u ,v \in \mathbb C$. We must have $v \neq 0$, for otherwise $H_2 \Omega =0$. Furthermore we have $\mathrm{tr}(\Omega) =(u - s -t) e_2$. Thus if we had $u - s -t=0$, submodule of $W$ generated by $\Omega$ would be a proper symplectic subspace, and hence couldn't contain a Lagrangian subspace. We conclude that $u- s -t \neq 0$. By passing to another vector related by the action of diagonal matrices in $\mathrm{Sp}_4$, we may put $u = 1 + s +t$ and $v=1$. The next step is to compute elements of $W$ obtained by acting on $\Omega$ with elements of $\mathfrak g$. We list them in the order of decreasing eigenvalue of $J$ (consecutive eigenvalues are $7, 5, 3, 3, 1, 1, -1 ,-1 , -3 , -5$):
\begingroup
\allowdisplaybreaks
\begin{subequations} \label{eq:sp4_lagrangian}
\begin{align*}
U_1 \Omega &= (t-s) e_1 \otimes e_{12}, \\
Y_{12} \Omega &= (2 + 2s + 3t) e_2 \otimes e_{12} + e_1 \otimes e_{14}, \\
U_2 \Omega &= e_2 \otimes e_{14} + e_4 \otimes e_{12}, \\
X_{12} \Omega &= (1+2s +t) e_1 \otimes \eta - 2 (1+s +t) e_2 \otimes e_{14} - t e_4 \otimes e_{12}, \\
H_1 \Omega &= e_4 \otimes e_{14}, \\
H_2 \Omega &= s e_1 \otimes e_{23} + t e_3 \otimes e_{12} + (1+s +t) e_2 \otimes \eta -2 e_4 \otimes e_{14}, \tag{\ref{eq:sp4_lagrangian}} \\
V_2 \Omega &= -s e_1 \otimes e_{34} + t e_3 \otimes e_{14} + (1+s+t) e_4 \otimes \eta, \\
X_{21} \Omega &= (2+3s + 2 t) e_2 \otimes e_{23} - e_3 \otimes e_{14} - e_4 \otimes \eta, \\
Z_{12} \Omega &= s e_4 \otimes e_{23} + (1+s+2t) e_3 \otimes \eta - 2(1+s+t) e_2 \otimes e_{34}, \\
V_1 \Omega &= (s-t) e_3 \otimes e_{23} + e_4 \otimes e_{34}.
\end{align*}
\end{subequations}%
\endgroup
The only nontrivial scalar products between vectors listed above are\footnote{This calculation is greatly facilitated by the fact that eigenvectors of $J$ are orthogonal unless their eigenvalues add up to zero.}:
\begin{subequations} \label{eq:sp4_scalar_products}
\begin{gather*}
\omega(Y \Omega, V_1 \Omega) = \omega(U_2 \Omega, Z_{12} \Omega) = \omega(H_2 \Omega, X_{21} \Omega)  = \omega( X_{21} \Omega,H_1 \Omega) = 4 + 6s + 4t,  \\
\omega(X_{12} \Omega, Z_{12} \Omega) = \omega(H_2 \Omega, V_2 \Omega) = 4 + 8s + 8t + 4 s^2 + 12 st + 4t^2. \tag{\ref{eq:sp4_scalar_products}}
\end{gather*}
\end{subequations}
All these scalar products vanish if and only if $(s,t)$ is chosen as $\left( - \frac{4}{5}, \frac{1}{5} \right)$ or $\left(0, -1 \right)$. Vectors (\ref{eq:sp4_lagrangian}) are linearly independent in both cases. This means that we have found two Airy data, with $\Omega$ of one of the following forms:
\begin{subequations}
\begin{align}
\Omega_1 &= - \frac{4}{5} e_1 \otimes e_{23} + \frac{1}{5} e_3 \otimes e_{12} + \frac{2}{5} e_2 \otimes \eta + e_4 \otimes e_{14}, \\
\Omega_2 &= - e_3 \otimes e_{12} + e_4 \otimes e_{14}.
\end{align}
\end{subequations}
Now let $p$ be the projection onto $F^{\perp} \subseteq W$. We have
\begin{subequations}
\begin{align}
p(\Omega_1) &= - \frac{7}{15} e_1 \otimes e_{23} + e_2 \otimes \left( \frac{2}{5} e_{13} - \frac{1}{15} e_{24} \right) + \frac{8}{15} e_3 \otimes e_{12} + e_4 \otimes e_{14}, \\
p( \Omega_2) &= \frac{1}{3} e_1 \otimes e_{23} + \frac{1}{3} e_2 \otimes e_{24} - \frac{2}{3} e_3 \otimes e_{12} + e_4 \otimes e_{14}.
\end{align}
\label{eq:sp4_projected_omega}
\end{subequations}

We claim that Airy data $(\mathfrak g, W,  \Omega_1)$ and $(\mathfrak g, W, \Omega_2)$ are not isomorphic. Indeed, suppose that $(\phi, f) : (\mathfrak g, W, \Omega_1) \to (\mathfrak g, W,  \Omega_2)$ is an isomorphism. Every automorphism of $\mathfrak g$ is inner, so $\phi = \mathrm{Ad}_D$ for some $D \in \mathrm{Sp}_4$. Clearly $\mathrm{Ad}_D(J) =J$. Since $J$ is a diagonal matrix with distinct eigenvalues, this implies that $D$ is~diagonal. Origin of $W$ is the unique point where $\ell$ vanishes to second order, so~$\ell  = \phi^t \circ \ell \circ f$ implies that $f(0) =0$, i.e. $f$ is a linear map. On the other hand we have $\ell = \mathrm{Ad}_{D}^t \circ \ell \circ D$, so $T= D^{-1} \circ f $ is a $\mathfrak g$-module automorphism. Using Schur's lemma and (\ref{eq:sp4_projected_omega}) we see that no map of the form $f = D \circ T  $ with $T \in \mathrm{End}_{\mathfrak g}(W)$ carries $\Omega_1$ to $\Omega_2$, which completes the proof.

It follows from the previous paragraph that the orbits of $\Omega_1$ and $\Omega_2$ are distinct. To complete the computation of $\Sigma_s$, notice that any element of $\mathrm{Stab}(\Omega_i)$, $i=1,2$ commutes with $J$, so it has to be diagonal. Given this information, it is easy to check that stabilizers of $\Omega_1$ and $\Omega_2$ are trivial. Therefore we have
\begin{equation}
\Sigma_s \cong \mathrm{Sp}_4 \sqcup \mathrm{Sp}_4.
\end{equation}

We note that the element $J$ belongs to an $\mathfrak{sl}_2$ triple embedded in $\mathfrak g$, which is unique up to automorphism of the form $\mathrm{Ad}_D$ with diagonal $D \in \mathrm{Sp}_4$. For example we can put $J^+ = U_2 + X_{12}$ and $J^-= 4 V_2 + 3 X_{21}$. Then $[J,J^{\pm}]= \pm 2 J^{\pm}$, $[J^+, J^-]= J$. Regarding $\mathfrak g$ and $W$ as $\mathfrak{sl}_2$-modules, they decompose as
\begin{equation}
\mathfrak g \cong \mathfrak{sl}_2 \oplus \mathrm{Sym}^6 \mathbb C^2, \quad \quad W \cong \bigoplus_{j=0}^3 \mathrm{Sym}^{2j+1} \mathbb C^2.
\end{equation}

We remark that Airy data $\{ (\mathfrak g, W, \Omega_i) \}_{i=1}^2$ admit no nontrivial automorphisms. However they do admit a real structure $\sigma$ with $\mathfrak g^{\sigma} = \mathfrak{sp}_4(\mathbb R)$. Its construction is~analogous to that in the Subsection \ref{sec:sl2_clas}.

\subsection{Lie algebra $\mathfrak{sp}_{10}$}

The last simple Lie algebra to consider is $\mathfrak{sp}_{10}$. The only candidate for $W$ is $\Lambda_0^3 F$. We will use trace techniques to completely determine $J$. Once this is~done, we will find $\Omega$. Due to the large number of variables, calculations are difficult to carry out manually. We have performed them using symbolic algebra software.

Bases in $\mathfrak g$ and $F$ are chosen as for $\mathfrak{sp}_4$. Module $\Lambda_3 F$ has basis $\{ e_{ijk} \}_{1 \leq i < j < k \leq 5}$, where $e_{ijk} = e_i \wedge e_j \wedge e_k$. Symplectic form on $\Lambda^3 F$ is given by
\begin{equation}
\omega(e_{i_1 i_2 i_3}, e_{j_1 j_2 j_3}) = \sum_{\sigma \in S_3} \omega(e_{i_1}, e_{j_{\sigma(1)}}) \omega(e_{i_2}, e_{j_{\sigma(2)}}) \omega(e_{i_3}, e_{j_{\sigma(3)}}).
\end{equation}

Plugging the spectrum (\ref{eq:J_spectrum_formula}) into (\ref{eq:sp10_W_to_adj}) we obtain a system of five equations for five indeterminates $\{ \mathrm{tr} (\mathrm{ad}_J^{2k}) \}_{k=1}^5$. Its solution takes the form
\begin{subequations} \label{eq:sp10_trJ}
\begin{gather*}
\mathrm{tr} \left( \mathrm{ad}_J^2 \right) = 440, \quad \quad
\mathrm{tr} \left( \mathrm{ad}_J^4 \right) = \frac{24992}{3}, \quad \quad
\mathrm{tr} \left( \mathrm{ad}_J^6 \right) = \frac{16846720}{81}, \\
\mathrm{tr} \left( \mathrm{ad}_J^8 \right) = \frac{4329729536}{729}, \quad \quad
\mathrm{tr} \left( \mathrm{ad}_J^{10} \right) = \frac{133476300800}{729}. \tag{\ref{eq:sp10_trJ}}
\end{gather*}
\end{subequations}
Plugging this result into (\ref{eq:sp10_E_to_adj}) gives
\begin{subequations} \label{eq:sp10_EJ}
\begin{gather*}
E_2(J) = \frac{55}{3}, \quad \quad
E_4(J)  = \frac{2926}{27}, \quad \quad
E_6(J)  = \frac{172810}{729},  \\
E_8(J) = \frac{117469}{729}, \quad \quad
E_{10}(J) = \frac{1225}{729}. \tag{\ref{eq:sp10_EJ}}
\end{gather*}
\end{subequations}
Now expand $J= \sum_{i=1}^5 J^i H_i$ and consider the polynomial
\begin{equation}
\chi(t) = \prod_{i=1}^5 \left( t - (J^i)^2 \right) = t^5 + \sum_{k=1}^5 E_{2k}(J) (-t)^{5-k}.
\end{equation}
Its simple to check that roots of $\chi$ take the form $\{ \left( \frac{9-2j}{3} \right)^2 \}_{j=0}^4$. It follows that $J$ is~determined uniquely up to the action of the Weyl group to be
\begin{equation}
J = \frac{1}{3} \left( 9 H_1 + 7 H_2 + 5 H_3 + 3 H_4 + H_5 \right).
\label{eq:sp10_J}
\end{equation}
One can check that this element satisfies the spectral test (Proposition \ref{prop:J_spectrum}). We read off that $\Xi = \{ \mu_i \}_{i=1}^7$, where
\begin{subequations} \label{eq:sp10_mu_defns}
\begin{gather*}
\mu_1= (-1,1,1,0,0), \quad \quad
\mu_2 = (0,0,0,1,0), \quad \quad
\mu_3 = (1,-1,0,0,1),  \\
\mu_4 = (0,1,-1,0,1), \quad \quad
\mu_5 = (0,0,1,-1,1),  \\
\mu_6 = (1,0,-1,0,-1), \quad \quad
\mu_7 = (0,1,0,-1,-1), \tag{\ref{eq:sp10_mu_defns}}
\end{gather*}
\end{subequations}
in which $(c_1,c_2,c_3,c_4,c_5) = \sum_{i=1}^5 c_i L_i \in \mathfrak h^*$, with $L_i(H_j) = \delta_{ij}$.

It can be shown that vanishing of the component of $\Omega$ along $W_{\mu_i}$ with $i \in \{ 1, 2, 3\}$ implies that $\Omega$ is annihilated by some root vector of $\mathfrak g$. Furthermore the set $\Xi \setminus \{ \mu_4 \}$ doesn't span $\mathfrak h^*$, so vanishing of the component of $\Omega$ along $W_{\mu_4}$ would imply that $\Omega$ is annihilated by some diagonal element of $\mathfrak g$. Finally we observe that $\mu_5$ belongs to the linear span of $\{ \mu_i \}_{i=1}^4$, so $\Omega$ has to have nonzero component along $\mu_6$ or $\mu_7$. It~follows that up to action of diagonal elements of~$\mathrm{Sp}_{10}$ we must have
\begin{subequations} \label{eq:sp10_Omega}
\begin{gather*}
\Omega = e_{2,3,6} + e_{1,5,7} + e_{2,5,8} + a e_{3,5,9} + b e_{1,8,10} + c e_{2,9,10} \\
\alpha e_{4,1,6} + \beta e_{4,2,7} + \gamma e_{4,3,8} - (\alpha + \beta + \gamma) e_{4,5,10}. \tag{\ref{eq:sp10_Omega}}
\end{gather*}
\end{subequations}
where $(\alpha, \beta , \gamma) \neq (0,0,0)$ and $(b,c) \neq (0,0)$.

In principle there exists a huge amount of polynomials equations $\Omega$ has to satisfy in order for the subspace $\{ T \Omega \}_{T \in \mathfrak g} \subset W$ to be isotropic, but it turns out that only four of them are linearly independent. They take the form
\begin{subequations}
\begin{gather}
c =ab, \\
b=c (\alpha + 2 \beta + \gamma), \\
a (\alpha + \beta + 2 \gamma) = -1, \\
\alpha^2 + \beta^2 + \gamma^2 =- \alpha \beta - \beta \gamma - \gamma \alpha.
\end{gather}
\end{subequations}
The first equation implies that $b \neq 0$, for otherwise we would have $(b,c) = (0,0)$. Eliminating $c$ using the first equation we get that the second equation is inconsistent for $a=0$. Therefore also $c \neq 0$, by the first equation. Using the action of the diagonal elements of $\mathrm{Sp}_{10}$ again we may put $a=b=1$. Then $c=1$ also, by the first equation. The second and the third equation may be used to express $\alpha$ and $\beta$ as affine functions of $\gamma$:
\begin{equation}
\alpha = -3(\gamma+1), \quad \quad \beta = \gamma+2.
\end{equation}
Plugging this into the fourth equation we get a quadratic
\begin{equation}
6 \gamma^2 + 12 \gamma+7=0.
\end{equation}
Its solutions are complex conjugate and take the form
\begin{equation}
\gamma_{\pm} = \frac{1}{6} \left( - 6 \pm i \sqrt{6} \right).
\end{equation}
We denote vectors $\Omega$ corresponding to the two solutions by $\Omega_{\pm}$. By calculating ranks of certain matrices one may check that the subspaces $\{ T \Omega_{\pm} \}_{T \in \mathfrak g}$ have dimensions~$55$, and hence are Lagrangian. This means that $(\mathfrak g, W, \Omega_{\pm})$ is an Airy datum.

We ask if $(\mathfrak g, W, \Omega_{\pm})$ are isomorphic. Since the element $J$ is regular and $W$ is~irreducible, it is sufficient to check if there exists a diagonal element $D \in \mathrm{Sp}_{10}$ such that $D \Omega_+= \Omega_-$ or $D \Omega_+ = -\Omega_-$. This is a system of equations for the five independent matrix elements of $D$. One may show that no solution exists, so~that the two Airy structures are non-isomorphic. In particular $\Sigma_s$ has two connected components. This is striking, because there surely exists such isomorphism if we give up linearity over $\mathbb C$. In our standard bases it is given by the complex conjugation.

Finally, let's compute stabillizers of $\Omega_{\pm}$. Once again, since $J$ is regular, it is sufficient to consider diagonal elements of $\mathrm{Sp}_{10}$. It turns out that the equation $D \Omega_{\pm} = \Omega_{\pm}$ has three solutions, so $\mathrm{Stab}(\Omega_{\pm}) \cong \mathbb Z_3$. Therefore we have
\begin{equation}
\Sigma_s \cong \frac{\mathrm{Sp}_{10}}{\mathbb Z_3} \sqcup \frac{\mathrm{Sp}_{10}}{\mathbb Z_3}.
\end{equation}

\subsection{Summary}\label{ss:summary}

\begin{table}[ht]
\centering
\begin{tabular}{l|l|l|l|l}
$\mathfrak g$ & $W$  & $J$ & $\mathrm{Stab}(\Omega)$ & $\Sigma_s$   \\ \hline
$\mathfrak{sl}_2$ & $\mathrm{Sym}^5F$  & $\frac{1}{3} H$  &  $\mathbb Z_3$ & $\mathrm{SL}_2 / \mathbb Z_3$   \\
$\mathfrak{sl}_2$ & $\mathrm{Sym}^3F \oplus F$ & $H$ & $0$  & $\mathrm{SL}_2 \sqcup \mathrm{SL}_2$  \\
$\mathfrak{sp}_4$ & $F \otimes \widetilde F$  & $3H_1 + H_2$  & $0$  & $\mathrm{Sp}_4 \sqcup \mathrm{Sp}_4$  \\
$\mathfrak{sp}_4$ & $F \otimes \widetilde F$  & $3H_1 + H_2$  & $0$  & $\mathrm{Sp}_4 \sqcup \mathrm{Sp}_4$ \\
$\mathfrak{sp}_{10}$ & $\Lambda_0^3 F$ & $\frac{1}{3} \left( 9 H_1 + 7 H_2 + 5 H_3 + 3 H_4 + H_5 \right)$ & $\mathbb Z_3$ & $\frac{\mathrm{Sp}_{10}}{\mathbb Z_3} \sqcup \frac{\mathrm{Sp}_{10}}{\mathbb Z_3}$ \\
$\mathfrak{sp}_{10}$ & $\Lambda_0^3 F$ & $\frac{1}{3} \left( 9 H_1 + 7 H_2 + 5 H_3 + 3 H_4 + H_5 \right)$ & $\mathbb Z_3$ & $\frac{\mathrm{Sp}_{10}}{\mathbb Z_3} \sqcup \frac{\mathrm{Sp}_{10}}{\mathbb Z_3}$
\end{tabular}
\caption{List of isomorphism classes of Airy structures for simple Lie algebras. }
\label{tab:simple_Airy}
\end{table}

We have found that there exist six isomorphism classes of Airy data $(\mathfrak g, W , \Omega)$ with simple $\mathfrak g$. This is summarized in Table \ref{tab:simple_Airy}. It is worth noticing that $\mathfrak{sp}_4$ admits two non-isomorphic Airy data for which all invariant characteristics, such as the conjugacy class of $J$ or of the subgroup $\mathrm{Stab}(\Omega) \subseteq G$, coincide. This raises the question if there are other invariants which can be used to distinguish the two Airy data. Unfortunately, we have not found any. Similar statement applies to $\mathfrak{sp}_{10}$, but~in~this case lack of desired invariants is explained by the fact that the two Airy data are related by complex conjugation.

\section{Semisimple Lie algebras - examples} \label{sec:ss_examples}

\begin{proposition}
The number of isomorphism classes of indecomposable Airy data $(\mathfrak g, W, \Omega)$ with $\mathfrak g$~semisimple is countably infinite.
\end{proposition}
\begin{proof}
The upper bound follows from the fact that (up to isomorphism) there are countably many semisimple Lie algebras, each of which admits finitely many Airy data. We show that this bound is saturated by explicitly constructing an infinite family of mutually non-isomorphic indecomposable Airy data in the Subsection \ref{sec:infinite_family}.
\end{proof}

The only semisimple Lie algebra of rank $2$ which is not simple is $\mathfrak{sl}_2 \times \mathfrak{sl}_2$. Airy data for this algebra are classified in the Subsection \ref{sec:sl2sl2}. Already in this case the number of non-isomorphic Airy data is quite large. Therefore we don't carry out analogous computations for Lie algebras of rank $3$. Instead we will decide which algebras admit at least one Airy datum. This is facilitated by the following criterion.

Let $(\mathfrak g, W, \Omega)$ be an Airy datum with $\mathfrak g = \mathfrak g' \times \mathfrak g''$ semisimple. Let $G'$ be a~simply-connected Lie group with Lie algebra $\mathfrak g'$. Put $n' = \dim (\mathfrak g')$. $H^0(\mathfrak g', W)$ is a~symplectic submodule of $W$, so $W = W' \oplus H^0(\mathfrak g', W)$, where $W'$ is the orthogonal complement of $H^0(\mathfrak g', W)$. By construction $H^0(\mathfrak g', W')=0$. Let $\ell'$ be the moment map for the $\mathfrak g'$-action on $W'$ and define $\Sigma_s' = \{ q \in \ell'^{-1}(0) | \left. d \ell' \right|_{q} \text{ has rank } n'  \}$. We~have an~important inclusion
\begin{equation}
\Sigma_s \subseteq \Sigma_s' \times H^0(\mathfrak g', W').
\end{equation}
$\Sigma_s'$ is a coisotropic submanifold of $W'$ of dimension $2 \dim(W') - n' \geq n'$. This statement is particularly useful when this inequality is saturated. In this situation $\Sigma_s'$ is Lagrangian with a locally transitive $G'$-action. In particular there exists $J' \in \mathfrak g'$ such that $J' \Omega$ is the projection of $\Omega$ onto $W'$. More importantly, existence of an~Airy datum of the form $(\mathfrak g', W', \Omega')$ is necessary for $\Sigma_s \neq \emptyset$.

\begin{proposition}
Lie algebra $\mathfrak g_2 \times \mathfrak{sl}_2$ admits no Airy data.
\end{proposition}
\begin{proof}
Put $\mathfrak g' = \mathfrak g_2$. Inspection of the list of irreducible representations of $\mathfrak g_2$ shows that the only possible forms of $W'$ are $F^{\oplus 4}$ and $\mathfrak g_2^{\oplus 2}$, where $F$ is the unique irreducible representation of dimension $7$. Bound $\dim(W') \geq 2 \dim (\mathfrak g')$ is saturated in both cases, so $\Sigma_s' = \emptyset$ follows from the fact that $\mathfrak g_2$ admits no Airy data.
\end{proof}

It happens to be true that also $\mathfrak{sl}_3 \times \mathfrak{sl}_2$ doesn't admit any Airy data, but in this case more complicated argument, presented in the Subsection \ref{sec:sl3sl2}, is required. Lie~algebra $\mathfrak{sp}_4 \times \mathfrak{sl}_2$ admits three decomposable Airy data and at least one indecomposable --- see~Subsection~\ref{sec:sp4_sl2}. This exhausts the list of semisimple Lie~algebras with two simple factors and rank $3$. The only remaining Lie algebra of~rank~$3$ is~$\prod_{i=1}^3 \mathfrak{sl}_2$. This one does admit an indecomposable Airy datum\footnote{We don't claim that there exist no other indecomposable Airy data for this Lie algebra.} which is~a~special case of the construction presented in the Subsection \ref{sec:infinite_family}.

\subsection{Lie algebra $\mathfrak{sl}_2 \times \mathfrak{sl}_2$} \label{sec:sl2sl2}

We will now present the list of indecomposable Airy data for $\mathfrak g =\mathfrak{sl}_2 \times \mathfrak{sl}_2$, up~to isomorphism. We omit details of calculations, which are analogous to previous sections. Canonical generators of the first (resp. second) copy of $\mathfrak{sl}_2$ will be denoted by $H, X, Y$ (resp. $\widetilde H, \widetilde X, \widetilde Y$). Similarly, their fundamental modules are denoted by $F$ and $\widetilde F$. They are generated by $e_1, e_2$ and $\widetilde e_1, \widetilde e_2$, respectively.

\begin{enumerate}
\item $W = (\mathrm{Sym}^4 F\otimes \widetilde F) \oplus \widetilde F$. In this case $\Sigma_s$ has two connected components, which turn out to correspond to isomorphic Airy data. One may take $J = \frac{1}{2} H + \widetilde H$. Vectors $\Omega$ corresponding to the two connected components take the form
\begin{equation}
\Omega_{\pm} = \left( \pm \frac{i}{2} e_{1122} \otimes \widetilde e_1 + e_{1111} \otimes \widetilde e_2, \widetilde e_1 \right).
\end{equation}
Stabilizers of $\Omega_{\pm}$ are isomorphic to $\mathbb Z_4$ and we have
\begin{equation}
\Sigma_s \cong \left( \frac{\mathrm{PSL}_2}{\mathbb Z_2} \times \mathrm{SL}_2 \right) \sqcup \left( \frac{\mathrm{PSL}_2}{\mathbb Z_2} \times \mathrm{SL}_2 \right).
\end{equation}
We remark that this is the first example considered in this paper in which $\mathrm{Stab}(\Omega)$ is not contained in the one-parameter subgroup of $G$ generated by $J$. Another novel feature is that the element $J$ doesn't belong to any $\mathfrak{sl}_2$ triple embedded in $\mathfrak g$. Regarding the automorphism group, we have
\begin{equation}
\mathrm{Aut}(\mathfrak g, W, \Omega) = \mathrm{Inn}(\mathfrak g ,W , \Omega) \cong \mathbb Z_2.
\end{equation}
\item $W = (\mathrm{Sym}^4 F \otimes \widetilde F) \oplus F$. In this case $\Sigma_s \cong \mathrm{SL}_2 \times \mathrm{SL}_2$. One may take $J = H + 3 \widetilde H$,
\begin{equation}
\Omega = \left( \frac{1}{48} e_{1222} \otimes \widetilde e_1 + e_{1111} \otimes \widetilde e_2, e_1 \right).
\end{equation}
This Airy datum admits no nontrivial automorphisms.
\item $W = \mathrm{Sym}^3 F \otimes \mathrm{Sym}^2 \widetilde F$. One may take $J = \frac{1}{3} H + \frac{2}{3} \widetilde H$, $\Omega = e_{111} \otimes \widetilde e_{12} + e_{122} \otimes \widetilde e_{11}$. Stabilizer of $\Omega$ turns is generated by the commuting elements $e^{2\pi i J}$ and $e^{ i \pi \widetilde H }$, of order $3$ and $2$, respectively. Thus $\Sigma_s \cong \frac{\mathrm{SL}_2 \times \mathrm{PSL}_2}{\mathbb Z_3}$. Furthermore we have
\begin{equation}
\mathrm{Aut}(\mathfrak g, W, \Omega) = \mathrm{Inn}(\mathfrak g ,W , \Omega) \cong \mathbb Z_3
\end{equation}
\item $W = (F \otimes \widetilde F)^{\oplus 2} \oplus \mathrm{Sym}^3 F$ with $J = H$ and $\Omega = (e_1 \otimes \widetilde e_1,2 e_1 \otimes \widetilde e_2, e_{112})$. This is the only example in this paper in which $J$ is not a regular element. Stabilizer of $\Omega$ is trivial, so $\Sigma_s \cong \mathrm{SL}_2 \times \mathrm{SL}_2$. Nevertheless, the group of automorphisms of this Airy datum is nontrivial, $\mathrm{Aut}(\mathfrak g , W, \Omega) = \mathrm{AInn}(\mathfrak g, W, \Omega) \cong \mathbb Z_2$. More explicitly, let $\tau$ be the automorphisms of $W$ defined by
\begin{equation}
\tau(x,y,z) = (-y,x,z) \quad \quad \mathrm{for} \  x,y \in F \otimes \widetilde F, \ z \in \mathrm{Sym}^3 F.
\end{equation}
We have $\mathrm{Aut}(\mathfrak g, W, \Omega) = \{ \mathrm{id}, (\mathrm{Ad}_g, \tau g) \}$, where $g=e^{\pi \widetilde Y - \frac{\pi}{4} \widetilde X}$.
\item $W = (\mathrm{Sym}^2 F \otimes \widetilde F) \oplus (F \otimes \mathrm{Sym}^2 \widetilde F)$. One may take $J = H + \widetilde H$. Set $\Sigma_s$ has two connected components. Corresponding vectors $\Omega$ are of the form
\begin{subequations}
\begin{align}
\Omega_1 &= \left( e_{11} \otimes \widetilde e_2 , e_2 \otimes \widetilde e_{11} \right), \\
\Omega_2 &= \left( e_{11} \otimes \widetilde e_2 -4 e_{12} \otimes \widetilde e_1 , e_2 \otimes \widetilde e_{11} - 4 e_1 \otimes \widetilde e_{12} \right).
\end{align}
\end{subequations}
We have $\mathrm{Stab}(\Omega_1) \cong \mathbb Z_3$ and $\mathrm{Stab}(\Omega_2) =0$, so $\Sigma_s \cong \left( \frac{\mathrm{SL}_2 \times \mathrm{SL}_2}{\mathbb Z_3} \right) \sqcup \left( \mathrm{SL}_2 \times \mathrm{SL}_2 \right)$ and the
Airy data $A_1=(\mathfrak g , W , \Omega_1)$ and $A_2=(\mathfrak g, W, \Omega_2)$ are not isomorphic. Another difference between $A_1$ and $A_2$ is in the automorphism groups:
\begin{equation}
\mathrm{Aut}(A_1) \cong \mathbb Z_2 \ltimes \mathbb Z_3, \quad \quad \mathrm{Aut}(A_2) \cong \mathbb Z_2.
\end{equation}
More explicitly, let $z = (\mathrm{Ad}_g,g)$ with $g = \exp \left( \frac{2 \pi i}{3} (H-\widetilde H) \right)$ and $t = (\tau_{\mathfrak g}, \tau_W)$, where $\tau_{\mathfrak g}$ is the outer automorphism of $\mathfrak g$ which swaps the two simple factors and $\tau_W$ is the analogous automorphism of $W$. Elements $z,t$ generate $\mathrm{Aut}(A_1)$ and satisfy relations $z^3=t^2 = id$, $tz=z^{-1} t$. Group $\mathrm{Aut}(A_2)$ is generated~by~$t$. Generalization of this example is discussed in the Subsection \ref{sec:infinite_family}.
\end{enumerate}

Elements $\Omega$ in points $2-5$ have real coefficient, so they admit real structures $\sigma$ with $\mathfrak g^{\sigma} = \mathfrak{sl}_2(\mathbb R) \times \mathfrak{sl}_2(\mathbb R)$.

\subsection{Lie algebra $\mathfrak{sl}_3 \times \mathfrak{sl}_2$} \label{sec:sl3sl2}

In this subsection we consider the Lie algebra $\mathfrak{sl}_3 \times \mathfrak{sl}_2$. Let $F$, $\widetilde F$ be the defining representations of $\mathfrak{sl}_3$ and $\mathfrak{sl}_2$, respectively. We denote the standard basis of $\mathfrak{sl}_2$ by $\widetilde H, \widetilde X, \widetilde Y$. Since $\mathfrak{sl}_3$ admits no Airy data, we must have $\dim(H^0(\mathfrak{sl}_3,W)) < 6$. It~follows that the only possible forms of $W$ are
\begin{subequations} \label{eq:sl3_sl2_reps}
\begin{align*}
W_1 &= \left( (F \oplus F^*) \otimes \widetilde F \right) \oplus \mathfrak{sl}_3 \oplus \widetilde F,  \\
W_2 &= \left( (F \oplus F^*) \otimes \widetilde F \right) \oplus (F \oplus F^*) \oplus \widetilde R, \\
W_3 &= \left( (F \oplus F^*) \otimes \mathfrak{sl}_2 \right) \oplus \widetilde R, \tag{\ref{eq:sl3_sl2_reps}} \\
W_4 & = (\mathfrak{sl}_3 \otimes \widetilde F ) \oplus (F \oplus F^*),
\end{align*}
\end{subequations}
where $\widetilde R = \widetilde F^{\oplus 2}$ or $\widetilde R = \mathrm{Sym}^3\widetilde F$.


Consider the possibility that $W=W_1$. Since the projection of $\Omega$ onto $\widetilde F$ must be nonzero, $J$ is necessarily of the form $J = J' + \widetilde H$ with some $J' \in \mathfrak{sl}_3$. The~index of~$W$ considered as a representation of $\mathfrak{sl}_3$ is equal to $\frac{5}{3}$, so $\mathrm{tr}_W(J^2) = 14 + \frac{5}{3} \ \mathrm{tr}(\mathrm{ad}_{J'}^2)$. On the other hand $\mathrm{tr}_W(J^2) = 38 + 2\ \mathrm{tr}(\mathrm{ad}_{J'}^2)$, by Proposition \ref{prop:J_spectrum}. Comparing the two results we get $\mathrm{tr}(\mathrm{ad}_{J'}^2)=-72$, which is impossible. Hence $W=W_1$ is ruled out. Representation $W_2$ may also be excluded by similar reasoning. In this case we have $J=J' + \lambda \widetilde H$ with $\lambda \in \{ 1 , \frac{1}{3} \}$. Repetition of the calculation presented above gives $\mathrm{tr}(\mathrm{ad}_{J'}^2)<0$ in both cases. For $W_3$ value $\lambda = \frac{1}{3}$ is excluded for similar reasons, but~possibility of $\lambda=1$ remains. In this situation we have $\mathrm{tr}(\mathrm{ad}_{J'})^2 = 10 + \mathrm{tr}_{\widetilde R}(\widetilde H^2)$.

Using techniques described in the Appendix \ref{app:invariant_poly} we derive relations
\begin{subequations}
\begin{align}
\mathrm{tr}(\mathrm{ad}_T^4) &= \frac{1}{4} \mathrm{tr}(\mathrm{ad}_T^2)^2, \\
\mathrm{tr}_{W_3}(T^6) &= - \frac{1}{9} \mathrm{tr} \left( \mathrm{ad}_T^6 \right) + \frac{5}{324} \mathrm{tr}(\mathrm{ad}_T^2)^3.
\end{align}
\end{subequations}
for $T \in \mathfrak{sl}_3$. Combining this with the Proposition \ref{prop:J_spectrum} we get
\begin{equation}
\mathrm{tr} \left( \mathrm{ad}_{J'}^6 \right) =
\begin{cases}
- \frac{296872}{247} & \mathrm{for} \ \widetilde R = \widetilde F^{\oplus 2}, \\
- \frac{739800}{247} & \mathrm{for} \ \widetilde R = \mathrm{Sym}^3 \widetilde F.
\end{cases}
\end{equation}
This contradicts rationality of $J'$, so representation $W_3$ is excluded.

In the case of representation $W_4$ similar calculations give
\begin{equation}
\mathrm{tr}_{W_4}(T^6) = \frac{53}{27} \mathrm{tr}(\mathrm{ad}_T^6) + \frac{5}{972} \mathrm{tr}(\mathrm{ad}_T^2)^3
\end{equation}
for $T \in \mathfrak{sl}_3$. Plugging in $J = J' + \lambda \widetilde H$ we obtain
\begin{subequations}
\begin{align}
\mathrm{tr}(\mathrm{ad}_{J'}^2) &= 66, \\
896214+38880 \lambda^2 + \mathrm{tr}(\mathrm{ad}_{J'}^6) &=0.
\end{align}
\end{subequations}
Second relation is clearly inconsistent, so also $W_4$ is ruled out. Therefore $\mathfrak{sl}_3 \times \mathfrak{sl}_2$ admits no Airy structures.


\subsection{Lie algebra $\mathfrak{sp}_4 \times \mathfrak{sl}_2$} \label{sec:sp4_sl2}

In this subsection we shall confine ourselves to presenting a single example of an indecomposable Airy datum for the Lie algebra\footnote{The number of symplectic, $\mathfrak g$-modules of dimension $26$ is not large, so it is of course possible,
even if a bit tedious, to completely classify Airy data also in this case.} ${\mathfrak g} = \mathfrak{sp}_4 \times \mathfrak{sl}_2$.
We use the same notation for generators and bases relevant for the $\mathfrak{sp}_4$ algebra as in subsection \ref{sec:sp4_Airy} while, for the $\mathfrak{sl}_2$ algebra,
the notation used parallels the one in subsection \ref{sec:sl3sl2}.

The pertinent Airy datum exists for $J = H_1 + H_2 + \widetilde H$ and the module
\[
W = (\Lambda_0^2F\otimes \widetilde F)\oplus (F\otimes {\rm Sym}^2\widetilde F) \oplus {\rm Sym}^3 \widetilde F,
\]
where $F$ is the defining representations of $\mathfrak{sp}_4,$ $\widetilde F$ denotes  the defining representations of $\mathfrak{sl}_2$
(and {\bf not}, as in subsection \ref{sec:sp4_Airy}, the definig representation of $\mathfrak{so}_5$)
and $\Lambda^2_0 F$ denotes the codimension 1 subspace of these elements of $\Lambda^2 F$ whose any contraction with the symplectic form of $F$ vanishes.
Since
\begin{subequations}
\begin{align}
\mathrm{spec}_{\mathfrak g} ( \mathrm{ad}_J ) & = \{0,0,0,0,0,\pm 2,\pm 2,\pm 2,\pm 2\},
\\
\mathrm{spec}_{\rm W}(J) &= \{\pm 1,\pm 1,\pm 1,\pm 1,\pm 1,\pm 1,\pm 1,\pm 1,\pm 1,\pm 3,\pm 3,\pm 3,\pm 3\},
\end{align}
\end{subequations}
the spectral test is met. Denoting by $(c_1,c_2,\tilde c) = c_1 L_1 + c_2 L_2 + \widetilde c \widetilde L \in \mathfrak{h}^*$, where $L_i(H_j) = \delta_{i,j}$ and $\widetilde L (\widetilde H) = 1$,
we read of that  $\Xi = \{ \mu_i \}_{i=1}^9$ with
\begin{align*}
\mu_1 &= (1,1,-1), \quad \quad
\mu_2 = (1,-1,1), \quad \quad
\mu_3 = (-1,1,1),
\\
\mu_4 &=  (0,0,1), \quad \quad
\mu_5 = (1,0,0), \quad \quad
\mu_6 = (0,1,0),
\\
\mu_7 &= (-1,0,2), \quad \quad
\mu_8 = (0,-1, 2), \quad \quad
\mu_9 = (0,0,1).
\end{align*}
This gives
\begin{eqnarray}
\nonumber
\Omega
& = &
\left(
\alpha e_{12}\otimes {\widetilde e}_2 +\left(\beta_{+} e_{14} + \beta_{-} e_{23} + \beta_0 \eta\right)\otimes {\widetilde e}_1, 0,0
\right)
\\
& + &
\left(0, \left(a_1 e_1 +a_2 e_2\right)\otimes {\widetilde e}_{12} + \left(a_3 e_3 + a_4 e_4\right)\otimes {\widetilde e}_{11}, 0 \right)
\\
\nonumber
& + &
\left(
0,0,u {\widetilde e}_{112}
\right)
\end{eqnarray}
with complex parameters $\alpha, \beta_0,\beta_\pm, a_i,\ i = 1,\ldots 4$ and $u.$ Requiring that $\{T\Omega\}_{T\in {\mathfrak g}}$ is~an~isotropic subspace of $W$ we get the following set of equations:
\begin{subequations}
\begin{gather}
4\alpha \beta_0 = a_1 a_2, \hskip 1cm 4\alpha\beta_+ = -a_1^2, \hskip 1cm 4\alpha\beta_- = a_2^2, \label{lconst:set1}  \\
a_2(a_1-a_2) = 0  \hskip 1cm \beta_0^2 + \beta_+\beta_- - a_1a_3 - a_2a_4 = u^2. \label{lconst:set2}
\end{gather}
\end{subequations}
Using the group action generated by the five elements of ${\mathfrak g}$ commuting with $J,$
namely $H_1,H_2,X_{12},X_{21}$, and $\widetilde H$ and demanding that ${\dim }\{T\Omega\}_{T\in {\mathfrak g}} = 13$
(so that $\{T\Omega\}_{T\in {\mathfrak g}}$ is a Lagrangian subspace of $W$),
we obtain a single isomorphism class of Airy data. It can be represented by the vector
\begin{equation}
\Omega
=
\left(
\frac12 e_{12}\otimes {\widetilde e}_2 -\frac{1}{2}e_{14} \otimes {\widetilde e}_1,
e_1\otimes {\widetilde e}_{12} -e_3\otimes {\widetilde e}_{11} + e_4\otimes {\widetilde e}_{11},
{\widetilde e}_{112}
\right).
\end{equation}
\subsection{Lie algebra $\prod_{i=1}^N \mathfrak{sl}_2$} \label{sec:infinite_family}

In this subsection we shall consider  the Lie algebra ${\mathfrak g} = \prod_{i=1}^N {\mathfrak sl}_2^{(i)},$
where each ${\mathfrak sl}_2^{(i)}$ is an independent copy of the ${\mathfrak sl}_2$ algebra with defining representation
denoted by~$F^{(i)}.$ Standard generators of ${\mathfrak sl}_2^{(i)}$ will be denoted by $H^{(i)},  X^{(i)}$ and $Y^{(i)}$.
Analogous notation will be used for representations of ${\mathfrak sl}_2^{(i)}.$

There exists a family of Airy structures for ${\mathfrak g}$ where
\[
W = \bigoplus\limits_{i=1}^N \left( {\rm Sym}^2 F^{(i)}\otimes F^{(i+1)}\right),
\]
(the sum here is cyclic i.e.\ $N+1 \equiv 1$),
\(
J = \sum\limits_{i=1}^N H^{(i)},
\)
and
\[
\Omega = \frac{1}{\sqrt N}\sum\limits_{i=1}^N \left(e_{11}^{(i)} \otimes e^{(i+1)}_2 + 4\alpha_i e_{12}^{(i)}\otimes e_{1}^{(i+1)} \right).
\]
Vectors $H^{(i)}\Omega,  X^{(i)}\Omega$ and $Y^{(i)}\Omega$ turn out to be linearly independent and orthogonal with respect to the ${\mathfrak g}$
invariant symplectic form on $W$ if  $\alpha_i = -\alpha_{i-1}^2,\; i = 2,\ldots, N$ and $\alpha_1^{2^N-1} + 1 = 0.$

We shall discuss in detail the simplest case where all $\alpha_i = 0,$
i.e.
\[
\Omega = \frac{1}{\sqrt N}\sum\limits_{i=1}^N e_{11}^{(i)} \otimes e^{(i+1)}_2.
\]
Let $\hbox{\boldmath $q$} = {\rm e}^{\frac{2\pi \sqrt{-1}}{N}}$ and let us assume that $N$ is odd\footnote{The construction for even $N$ is completely analogous.}, $N = 2n+1.$
Let us also define:
\begin{equation}
Z_j = \sum\limits_{i=-n}^n \hbox{\boldmath $q$}^{ji}Z^{(i)}, \hskip 1cm j = -n, -n+1,\ldots,n,
\end{equation}
with $Z^{(i)} = H^{(i)}, X^{(i)}$ or $Y^{(i)}.$ These operators obey the algebra
\begin{equation}
[H_i,X_j] = 2X_{i\dotplus j}, \hskip 1cm [H_i,Y_j] = - 2Y_{i\dotplus j} \hskip .5cm {\rm and} \hskip .55cm [X_i,Y_j] =  H_{i\dotplus j}
\end{equation}
where $i \dotplus j = i + j\ {\rm mod}\ N.$

A Lagrangian complement of $T_\Omega \Sigma = {\rm lin}\{H_j\Omega, X_j\Omega,Y_j\Omega\}$ can be constructed as
\[
V = {\rm lin}\{H_j\overline{\Omega}, X_j\overline{\Omega},Y_j\overline{\Omega}\}
\]
where
\begin{equation}
\overline{\Omega} = \frac{1}{\sqrt N}\sum\limits_{i=-n}^n e_{22}^{(i)} \otimes e^{(i+1)}_1.
\end{equation}
It is the immediate to check that vectors
\begin{eqnarray*}
e^0_j & = & \frac{1}{2-\hbox{\boldmath $q$}^j}H_j\Omega \; = \; \frac{1}{\sqrt {N}} \sum\limits_{i=-n}^n \hbox{\boldmath $q$}^{ji}e_{11}^{(i)}\otimes e_2^{(i+1)},
\\
e^+_j & = & \hbox{\boldmath $q$}^{-j}X_j\Omega \; = \; \frac{1}{\sqrt{N}} \sum\limits_{i=-n}^n \hbox{\boldmath $q$}^{ji}e_{11}^{(i)}\otimes e_1^{(i+1)},
\\
e^-_j & = & \frac{1}{2} Y_j\Omega \; = \; \frac{1}{\sqrt{N}} \sum\limits_{i=-n}^n \hbox{\boldmath $q$}^{ji}e_{12}^{(i)}\otimes e_2^{(i+1)},
\end{eqnarray*}
and
\begin{eqnarray*}
f^0_{j} & = & \frac{1}{2} \frac{1}{2-\hbox{\boldmath $q$}^{-j}} H_{-j}\overline{\Omega} \; = \; -\frac{1}{2\sqrt{N}} \sum\limits_{i=-n}^n \hbox{\boldmath $q$}^{-ji} e_{22}^{(i)} \otimes e^{(i+1)}_1,
\\
f^+_{j} & = & \frac{1}{2} X_{-j}\overline{\Omega} \; = \; \frac{1}{\sqrt{N}}\sum\limits_{i=-n}^n \hbox{\boldmath $q$}^{-ji} e_{12}^{(i)} \otimes e^{(i+1)}_1,
\\
f^-_{j} & = & \frac{\hbox{\boldmath $q$}^j}{2} Y_{-j}\overline{\Omega} \; = \;\frac{1}{2\sqrt{N}}\sum\limits_{i=-n}^n \hbox{\boldmath $q$}^{-ji} e_{22}^{(i)} \otimes e^{(i+1)}_2,
\end{eqnarray*}
with $j =\{ -n , -n +1 , ..., n \}$ satisfy
\begin{equation}
\omega(e^a_j,f^b_k) =  \delta_{a+b}\delta_{j,k}.
\end{equation}
Decomposing
\[
W \ni w = \sum\limits_{j=-n}^n\sum\limits_{a=0,\pm}\left(\alpha^a_{-j}e^a_j + \beta^a_{j} f^a_j\right)
\]
we get
\begin{eqnarray}
\nonumber
{\mathcal X}_k
& \equiv &
\omega(X_k w, \Omega) + \frac12\omega(X_kw,w)
\\ \nonumber
& = &
\hbox{\boldmath $q$}^k\,\beta^-_{k}
+
\sum\limits_{l=-n}^n
\left(
\hbox{\boldmath $q$}^k\alpha_{-l}^0\beta_{l\dotplus k}^-+ \alpha_{-l}^-\beta_{l\dotplus k}^0 + \beta_l^+\beta_{-l\dotplus k}^--\frac12 \hbox{\boldmath $q$}^k \alpha_{-l}^-\alpha_{l\dotplus k}^-
\right),
\\
{\mathcal H}_k
& \equiv &
\omega(H_k w, \Omega) + \frac12\omega(H_kw,w)
\\ \nonumber
& = &
(2-\hbox{\boldmath $q$}^k)\,\beta^0_{k}
+
\sum\limits_{l=-n}^n
\left(
\left(2-\hbox{\boldmath $q$}^k\right)\alpha_{-l}^0\beta_{l\dotplus k}^0
+
\left(2+\hbox{\boldmath $q$}^k\right)\alpha_{-l}^+\beta_{l\dotplus k}^-
-
\hbox{\boldmath $q$}^k\alpha_{-l}^-\beta_{l\dotplus k}^+
\right),
\\
\nonumber
{\mathcal Y}_k
& \equiv &
\omega(Y_k w, \Omega) + \frac12\omega(Y_kw,w)
\\ \nonumber
& = &
2\,\beta^+_{k}
+
\sum\limits_{l = -n}^n
\left(
2\alpha_{-l}^0\beta_{l\dotplus k}^+
+
\hbox{\boldmath $q$}^k\alpha_{-l}^+\beta_{l\dotplus k}^0
+
\frac12\hbox{\boldmath $q$}^k \beta_l^+\beta_{-l\dotplus k}^+
-
2\alpha_{-l}^+\alpha_{l\dotplus k}^-
\right).
\end{eqnarray}
To express the hamiltonians above in the Airy form, we denote
\[
\hbox{\boldmath $q$}^k\,\beta^-_{k} = y_k^+,
\hskip 5mm
(2-\hbox{\boldmath $q$}^k)\,\beta^0_{k} = y_k^0,
\hskip 5mm
2\,\beta^+_{k} = y_k^-
\]
and introduce a set of variables $x^a_k$ ``conjugate'' to $y^a_k$
\[
\alpha_k^+ = - \hbox{\boldmath $q$}^k\,x_k^-,
\hskip 5mm
\alpha_k^0 = - (2-\hbox{\boldmath $q$}^k)\,x_k^0,
\hskip 5mm
\alpha_k^- = -2\, x_k^+
\]
such that the Poison bracket  on the space ${\rm lin}\{\alpha^a_k,\beta^a_k\} = {\rm lin}\{x^a_k,y^a_k\}$ is preserved:
\begin{equation}
\left\{y_k^a,x_l^b\right\} = \left\{\alpha^a_k,\beta^b_l\right\} = \delta_{a+b,0}\delta_{k+l,0}.
\end{equation}
This gives
\begin{eqnarray}
\label{sl2n:final:hamiltonians}
\nonumber
{\mathcal X}_k
& = &
y_k^+
-
\sum\limits_{l=-n}^n
\left(
\left(2\hbox{\boldmath $q$}^{-l}-1\right)x_{-l}^0y_{l\dotplus k}^+
+
\frac{2}{2-\hbox{\boldmath $q$}^{k+l}} x_{-l}^+y_{l\dotplus k}^0
+
2\hbox{\boldmath $q$}^k x_{-l}^+x_{l\dotplus k}^+
-
\frac{1}{2}\hbox{\boldmath $q$}^{l-k} y_l^-y_{-l\dotplus k}^+
\right),
\\
{\mathcal H}_k
& = &
y_k^0
-
\sum\limits_{l=-n}^n
\left(
\frac{\left(2-\hbox{\boldmath $q$}^k\right)\left(2-\hbox{\boldmath $q$}^l\right)}{2-\hbox{\boldmath $q$}^{k+l}} x_{-l}^0y_{l\dotplus k}^0
+
\left(2\hbox{\boldmath $q$}^{-k}+1\right)x_{-l}^-y_{l\dotplus k}^+
-
\hbox{\boldmath $q$}^kx_{-l}^+y_{l\dotplus k}^-
\right),
\\
\nonumber
{\mathcal Y}_k
& = &
y_k^-
-
\sum\limits_{l = -n}^n
\left(
\left(2-\hbox{\boldmath $q$}^{l}\right)x_{-l}^0 y_{l\dotplus k}^-
+
\frac{\hbox{\boldmath $q$}^{l+k}}{2-\hbox{\boldmath $q$}^{l+k}} x_{-l}^-y_{l\dotplus k}^0
+
4\hbox{\boldmath $q$}^l x_{-l}^- x_{l\dotplus k}^+
-
\frac18 \hbox{\boldmath $q$}^ky_l^-y_{-l\dotplus k}^-
\right).
\end{eqnarray}
It is not difficult to check explicitly that hamiltonians (\ref{sl2n:final:hamiltonians}) do satisfy the algebra
\begin{equation}
\left\{{\mathcal H}_k,{\mathcal X}_m\right\} = 2{\mathcal X}_{k\dotplus m},
\hskip 10mm
\left\{{\mathcal H}_k,{\mathcal Y}_m\right\} = -2{\mathcal Y}_{k\dotplus m},
\hskip 10mm
\left\{{\mathcal X}_k,{\mathcal Y}_m\right\} = {\mathcal H}_{k\dotplus m}.
\end{equation}

\section{Outlook} \label{sec:conclusions}

Many new Airy structures were found in this work. It is left for future studies to~find their partition functions, or at least discuss their properties. It would be particularly interesting to find a relation between them and fields in which topological recursion has found applications, or with some quantum systems studied in physics. Perhaps that could shed some light on the striking fact that Airy structures for semisimple Lie algebras are so much constrained. We believe that derivation of integral representations of partition functions could be~particularly illuminating.

Having classified Airy structures for simple Lie algebras, it is natural to ask for an~extension to semisimple Lie algebras. As shown by presented examples, in~this~case the number of~distinct Airy structures is infinite. However it is finite for any given semisimple Lie algebra, so it could be that this problem is~manageable. Some difficulties do arise, though. Firstly, there are many Lie~algebras and representations to consider. It is not clear to us how to generate a~complete list. Secondly, for a given representation of a Lie algebra of high rank the number of cases one has to consider in order to find the possible forms of $J$ is~large. We avoided this step in the derivation of the $\mathfrak{sp}_{10}$ Airy structures by~deriving the only consistent form of $J$ directly from its spectral properties and relations between invariant polynomials. This method typically breaks down for Lie algebras with more than one simple factor. Indeed, each simple factor contributes its own set of invariant polynomials, so we get more unknowns than equations to solve. Some new restrictions on $J$ would have to be derived in order to make this method viable.

Last but not least, it would be interesting to partially extend our results to~more general classes of Airy structures. Besides allowing more general Lie~algebras, one~could also consider Lie superalgebras with semisimple even part. If it is possible to generalize some of our findings to~infinite-dimensional Airy structures, that could have direct consequences for~classical topological recursion. Kac-Moody algebras generalize simple Lie~algebras in a natural way and have direct connections with conformal field theory and integrable systems, so they would be interesting to study in this context.


\section*{Acknowledgments}
We would like to thank G. Borot for discussion.
Calculations presented in this paper were facilitated by the use of Mathematica package LieART \cite{LieART}.
The work of BR was supported by the Faculty of Physics, Astronomy and Applied Computer Science grant MSN 2019 (N17/MNS/000040) for young scientists and PhD students.
The work of LH was supported by the TEAM programme of the Foundation for Polish Science co-financed by the European Union under the European Regional Development Fund (POIR.04.04.00-00-5C55/17-00).


\appendix

\section{Lie algebra cohomology} \label{app:cohom}

We present an \textit{ad hoc} definition of the first two Lie algebra cohomology groups, sufficient for our purposes. For a more conceptual treatment of the subject, see \cite{HiltonStammbach}.

Let $\mathfrak g$ be a Lie algebra and $M$ - a representation of $\mathfrak g$. The zeroth cohomology group $H^0(\mathfrak g,M)$ of $\mathfrak g$ valued in $M$ is defined as the space of all element of $M$ annihilated by $\mathfrak g$. To define the first cohomology group, we let $Z^1(\mathfrak g, M)$ be the vector space of linear maps (called cocycles) $\gamma : \mathfrak g \to M$ such that $\gamma([T,S])=T \gamma(S) - S \gamma (T)$, and $B^1(\mathfrak g, M)$ the space of linear maps (called coboundaries) $\mathfrak g \to M$ of the form $\gamma(T) = T m$ for some $m \in M$. Every coboundary is a cocycle, so it makes sense to put $H^1(\mathfrak g,M) = \frac{Z^1(\mathfrak g,M)}{B^1(\mathfrak g, M)}$. The following fact is used in this work:

\begin{proposition}[Whitehead]
Let $\mathfrak g$ be a finite-dimensional semisimple Lie algebra and $M$ - a finite-dimensional $\mathfrak g$-module. Then $H^1(\mathfrak g, M)=0$.
\end{proposition}


\section{Semisimple and regular elements} \label{app:ss_regular}

Let $W$ be a finite-dimensional vector space and let $T \in \mathrm{End}(W)$. We say that $T$ is~semisimple if for every $T$-invariant subspace $V \subseteq W$ there exists a~complementary $T$-invariant subspace $V'$, so that $W = V \oplus V'$. Since we restrict attention to vector spaces over $\mathbb C$, operator $T$ is semisimple if and only if it is diagonalizable.

\begin{proposition}[Jordan-Chevalley decomposition]
Let $T$ be a linear operator on a~finite-dimensional vector space $W$. There exist unique linear operators $T_{ss}, T_n$ on~$W$ such that $T_{ss}$ is semisimple, $T_{n}$ is nilpotent, $T_{ss}T_n = T_n T_{ss}$ and $T=T_{ss}+T_n$. Furthermore there exist polynomials $p,q \in \mathbb C[t]$ such that $T_{ss}=p(T)$, $T_n=q(T)$. In~particular every $T$-invariant subspace of $W$ is $T_{ss}$- and $T_n$-invariant.
\end{proposition}

\begin{proposition}
Let $\mathfrak g$ be a finite-dimensional, semisimple Lie algebra. For any $T \in \mathfrak g$ there exist unique $T_{ss}, T_n \in \mathfrak g$ such that $(\mathrm{ad}_T)_{ss} = \mathrm{ad}_{T_{ss}}$, $(\mathrm{ad}_T)_{n} = \mathrm{ad}_{T_{n}}$. Moreover $T_{ss}$ (resp. $T_n$) acts as a~semisimple (resp. nilpotent) operator in every finite-dimensional $\mathfrak g$-module.
\end{proposition}

Let $\mathfrak g$ be a finite-dimensional, semisimple Lie algebra. Element $T \in \mathfrak g$ is said to be semisimple (resp. nilpotent) if $T=T_{ss}$ (resp. $T=T_n$). Set of semisimple elements of $\mathfrak g$ is nonempty and Zariski open, hence dense. It coincides with the union of all Cartan subalgebras of $\mathfrak g$.

Rank of $\mathfrak g$ is defined as the greatest integer $r$ such that the characteristic polynomial of $\mathrm{ad}_T$ vanishes at zero with multiplicity at least $r$ for every $T \in \mathfrak g$. Element $T \in \mathfrak g$ is said to be regular if its characteristic polynomial vanishes at zero with multiplicity exactly $r$. By construction, the set of regular elements of $\mathfrak g$ is~nonempty and Zariski open. One can show that it is contained in the set of~semisimple elements. If $T \in \mathfrak g$ is a regular element, then the commutant $\{ T' \in \mathfrak g | [T,T']=0 \}$ is the unique Cartan subalgebra of $\mathfrak g$ which contains $T$. Now~suppose that some Cartan subalgebra $\mathfrak h \subseteq \mathfrak g$ is chosen. Element $T \in \mathfrak h$ is~regular in $\mathfrak g$ if and only if $\alpha(T) \neq 0$ for every root $\alpha$.

We remark that some authors define $T$ to be regular if the dimension of its commutant is equal to $r$. Elements with this property are not necessarily semisimple. However, the two notions do coincide for semisimple elements.



\section{Invariant polynomials on $\mathfrak{sp}_{10}$} \label{app:invariant_poly}

Let $G$ be a complex semisimple Lie group with Lie algebra $\mathfrak g$. Choose a Cartan subalgebra $\mathfrak h \subset \mathfrak g$ and let $\mathcal W$ be the corresponding Weyl group. Denote the algebra of $G$-invariant polynomial functions on $\mathfrak g$ by $\mathbb C[\mathfrak g]^G$ and the algebra of $\mathcal W$-invariant polynomial functions on $\mathfrak h$ by $\mathbb C[\mathfrak h]^{\mathcal W}$. If $\phi \in \mathbb C[\mathfrak g]^G$, then the restriction $\left. \phi \right|_{\mathfrak h}$ belongs to $\mathbb C[\mathfrak h]^{\mathcal W}$. In other words, we have a homomorphism
\begin{equation}
\mathrm{res} : \mathbb C[\mathfrak g]^G \ni \phi \mapsto \left. \phi \right|_{\mathfrak h} \mathbb C[\mathfrak h]^{\mathcal W}.
\end{equation}
We claim that $\mathrm{res}$ is injective. Indeed, suppose that $\left. \phi \right|_{\mathfrak h}=0$. If $T \in \mathfrak g$ is~semisimple, then the $G$-orbit of $T$ intersects $\mathfrak h$ nontrivially, so $\phi(T)=0$. Since the set of~semisimple elements is dense and $\phi$ is continuous, we must have $\phi=0$.

\begin{proposition}[Chevalley]
Homomorphism $\phi$ is also surjective.
\end{proposition}
\begin{proof}
See \cite[p.~126-128]{Humphreys}.
\end{proof}

Now let's specialize to $\mathfrak g = \mathfrak{sp}_{10}$. We use the standard choice of $\mathfrak h$, basis in $\mathfrak g$ and basis is $\mathfrak h^*$ described in \cite{FultonHarris}. Weyl group takes the form
\begin{equation}
\mathcal W = S_5 \ltimes \mathbb Z_2^5,
\end{equation}
with $S_5$ acting on $\{ L_i \}_{i=1}^5$ by permutations and $\mathbb Z_2^5$ generated by the five reflections $L_i \mapsto - L_i$. It follows that elements of $\mathbb C[\mathfrak h]^{\mathcal W}$ are symmetric polynomials in $\{ L_i^2 \}_{i=1}^5$. Therefore by the fundamental theorem of symmetric polynomials \cite{Macdonald}, functions
\begin{equation}
E_{2k} : \ \mathfrak h \ni \sum_{i=1}^5 T^i H_i \mapsto \sum_{1 \leq i_1 < ... < i_k \leq 5} (T^{i_1})^2 ... (T^{i_k})^2 \in \mathbb C
\end{equation}
with $k \in \{ 1, ..., 5 \}$ are algebraically independent generators of $\mathbb C[\mathfrak h]^{\mathcal W}$. Hence they extend uniquely to invariant polynomials on $\mathfrak g$ and we have
\begin{equation}
\mathbb C[\mathfrak g]^G = \mathbb C[E_2, E_4,E_6,E_8,E_{10}].
\end{equation}
In particular the dimension of the space of invariant polynomials on $\mathfrak g$ of degree $2k$ is equal to the number of partitions of $k$.

Products of $\{ E_{2k} \}_{k=1}^5$ furnish a basis in $\mathbb C[\mathfrak g]^G$. It will be useful to construct several other bases. Let $V$ be a representation of $\mathfrak g$. Define
\begin{equation}
Q_{2k}^V : \ \mathfrak g \ni T \mapsto \mathrm{tr}_V \left( T^{2k} \right) \in \mathbb C
\end{equation}
for $k \in \mathbb N$. We have $Q_{2k}^V \in \mathbb C [\mathfrak g]^G$. Consider first the adjoint representation, $V= \mathfrak g$. By the preceding discussion there exist coefficients $\{ \alpha_i \}_{i=1}^{18}$ such that
\begin{subequations} \label{eq:Q_to_E}
\begin{align*}
Q_2^{\mathfrak g} &= \alpha_1 E_2, \\
Q_4^{\mathfrak g} &= \alpha_2 E_4 + \alpha_3 E_2^2, \\
Q_6^{\mathfrak g} &= \alpha_4 E_6 + \alpha_5 E_4 E_2 + \alpha_6 E_2^3, \tag{\ref{eq:Q_to_E}}  \\
Q_8^{\mathfrak g} &= \alpha_7 E_8 + \alpha_8 E_6 E_2 + \alpha_9 E_4^2 + \alpha_{10} E_4 E_2^2 + \alpha_{11} E_2^4, \\
Q_{10}^{\mathfrak g} &= \alpha_{12} E_{10} + \alpha_{13} E_8 E_2 + \alpha_{14} E_6 E_4 + \alpha_{15} E_6 E_2^2 + \alpha_{16} E_4^2 E_2 + \alpha_{17}E_4 E_2^3 + \alpha_{18} E_2^5.
\end{align*}
\end{subequations}
Values of coefficients $\alpha_i$ may be found by evaluating this equation on any sufficiently large set of elements of $\mathfrak h$ and solving a system of linear equations. Their exact values will be of no use for us, but by computing them allows to check that $\{ Q_{2k}^{\mathfrak g} \}_{k=1}^5$ generate the algebra $\mathbb C[\mathfrak g]^G$. By dimensionality reasons they have to be algebraically independent, so we have
\begin{equation}
\mathbb C[\mathfrak g]^G = \mathbb C[Q_2^{\mathfrak g}, Q_4^{\mathfrak g},Q_6^{\mathfrak g},Q_8^{\mathfrak g},Q_{10}^{\mathfrak g}].
\label{eq:adjoint_traces_generate}
\end{equation}

Knowing that (\ref{eq:adjoint_traces_generate}) holds, we are guaranteed that the polynomials $E_{2k}$ and~$Q_{2k}^V$ may be expressed as polynomials in $\{ Q_{2k} \}_{k=1}^5$. Coefficients of these expansion are~useful and may be derived as in the previous paragraph. Firstly,
\begin{subequations} \label{eq:sp10_E_to_adj}
\begin{align*}
E_2 = &\frac{1}{24} Q_2^{\mathfrak g}, \\
E_4 =& - \frac{1}{72} Q_4^{\mathfrak g} + \frac{1}{864} (Q_2^{\mathfrak g})^2, \\
E_6 =& \frac{1}{252} Q_6^{\mathfrak g} - \frac{31}{36288} Q_4^{\mathfrak g} Q^{\mathfrak g}_2 + \frac{13}{435456} ( Q^{\mathfrak g}_2 )^3, \tag{\ref{eq:sp10_E_to_adj}} \\
E_8 = &- \frac{1}{1104} Q_8^{\mathfrak g} + \frac{5}{23184} Q^{\mathfrak g}_6 Q^{\mathfrak g}_2 + \frac{139}{715392} (Q^{\mathfrak g}_4)^2 \\
&- \frac{2111}{60092928} Q^{\mathfrak g}_4 (Q^{\mathfrak g}_2)^2 + \frac{1115}{1442230272} (Q^{\mathfrak g}_2)^4, \\
E_{10} =& \frac{1}{5220} Q^{\mathfrak g}_{10} - \frac{11}{256128} Q^{\mathfrak g}_8 Q^{\mathfrak g}_2 - \frac{19}{175392} Q^{\mathfrak g}_6 Q^{\mathfrak g}_4 + \frac{1}{123648} Q^{\mathfrak g}_6 (Q^{\mathfrak g}_2)^2, \\
&+\frac{18799}{1161796608} (Q^{\mathfrak g}_4)^2 Q^{\mathfrak g}_2 - \frac{1931}{1549062144} Q^{\mathfrak g}_4 (Q^{\mathfrak g}_2)^3 + \frac{33449}{1672987115520} (Q^{\mathfrak g}_2)^5
\end{align*}
\end{subequations}
Now let $W$ be the set of all elements of $\Lambda^3 F$ whose contraction with the symplectic form vanishes. $W$ is an irreducible representation of $\mathfrak g$ of dimension $110$. It will be important to have an expression for $\{ Q_{2k}^{W} \}_{k=1}^5$ in terms of $\{ Q_{2k}^{\mathfrak g} \}_{k=1}^5$:
\begin{subequations} \label{eq:sp10_W_to_adj}
\begin{align*}
Q_2^{W} = & \frac{9}{4} Q_2^{\mathfrak g}, \\
Q_4^W =& - \frac{1}{2} Q_4^{\mathfrak g} + \frac{13}{96} (Q_2^{\mathfrak g})^2, \\
Q_6^W =& -\frac{11}{14} Q_6^{\mathfrak g} + \frac{55}{1008} Q_4^{\mathfrak g} Q^{\mathfrak g}_2 + \frac{365}{48384} ( Q^{\mathfrak g}_2 )^3, \tag{\ref{eq:sp10_W_to_adj}} \\
Q_8^W = & \frac{317}{46} Q_8^{\mathfrak g} - \frac{1421}{828} Q^{\mathfrak g}_6 Q^{\mathfrak g}_2 - \frac{20755}{14904} (Q^{\mathfrak g}_4)^2 + \frac{92365}{357696} Q^{\mathfrak g}_4 (Q^{\mathfrak g}_2)^2 - \frac{163975}{34338816} (Q^{\mathfrak g}_2)^4, \\
Q_{10}^W =& \frac{1623}{58} Q^{\mathfrak g}_{10} - \frac{49512}{10672} Q^{\mathfrak g}_8 Q^{\mathfrak g}_2 - \frac{5605}{348} Q^{\mathfrak g}_6 Q^{\mathfrak g}_4 + \frac{8785}{13248} Q^{\mathfrak g}_6 (Q^{\mathfrak g}_2)^2,  \\
&+\frac{2256725}{1152576} (Q^{\mathfrak g}_4)^2 Q^{\mathfrak g}_2 - \frac{2885975}{27661824} Q^{\mathfrak g}_4 (Q^{\mathfrak g}_2)^3 + \frac{3515225}{2655535104} (Q^{\mathfrak g}_2)^5
\end{align*}
\end{subequations}

\end{document}